\documentclass[10pt]{article}
\usepackage{amsmath,amssymb,mathrsfs}
\usepackage{subfigure}
\usepackage{graphicx,color,url}
\usepackage{epsfig}
\usepackage{algorithm}
\usepackage{algorithmic}
\usepackage{amsthm}

\newtheorem{theorem}{Theorem}[section]
\newtheorem{corollary}[theorem]{Corollary}

\newtheorem{definition}[theorem]{Definition}
\newtheorem{remark}[theorem]{Remark}
\newtheorem{lemma}[theorem]{Lemma}

\newtheorem{proposition}{Proposition}[section]

\begin{document}

\title{Analysis-suitable T-splines: characterization, refineability,
  and approximation}

\author{Xin Li \\
School of mathematical science, USTC, \\
Hefei, Anhui Province 230026, P. R. China\\
lixustc@ustc.edu.cn\\
M. A. Scott\\
Department of Civil and Environmental Engineering, \\
Brigham Young
University, Provo, UT 84602, USA\\
michael.scott@byu.edu }

\maketitle

\begin{abstract}
We establish several fundamental properties of analysis-suitable
T-splines which are important for design and analysis. First, we
characterize T-spline spaces and prove that the
space of smooth bicubic polynomials, defined over the extended T-mesh of an
analysis-suitable T-spline, is contained in the corresponding
analysis-suitable T-spline space. This is accomplished through the
theory of perturbed analysis-suitable T-spline
spaces and a simple topological dimension formula. Second, we
establish the theory of analysis-suitable local refinement and describe
the conditions under which two analysis-suitable T-spline spaces are
nested. Last, we demonstrate that these results can be used to
establish basic approximation results which are critical for analysis.
\end{abstract}


\section{Introduction}
T-splines were originally introduced in Computer Aided Design (CAD) as
a superior alternative to NURBS~\cite{SeZhBaNa03} and have since emerged as an important
technology across several disciplines including industrial,
architectural, and engineering design, manufacturing, and engineering
analysis. T-splines can model complicated designs as a single,
watertight geometry and can be locally
refined~\cite{SeCaFiNoZhLy04,ScLiSeHu10}. These basic properties
make it possible to merge multiple NURBS patches into a single
T-spline~\cite{Ip05,SeZhBaNa03} and any trimmed NURBS model can be
represented as a watertight T-spline~\cite{Sederberg08}.

The use of T-splines as a basis for isogeometric analysis has gained
widespread attention~\cite{Bazilevs2009,ScBoHu10,ScLiSeHu10,Verhoosel:2010vn,Verhoosel:2010ly,BoScLaHuVe11,BeBaDeHsScHuBe09,SchDeScEvBoRaHu12,ScSiEvLiBoHuSe12}. Isogeometric
analysis was introduced in~\cite{HuCoBa04} and described in detail
in~\cite{Cottrell:2009rp}. The isogeometric paradigm is simple: use the smooth spline basis
that defines the geometry as the basis for analysis. Traditional
design-through-analysis procedures such as geometry clean-up,
defeaturing, and mesh generation are simplified or eliminated entirely. Additionally, the
higher-order smoothness provides substantial gains to analysis in
terms of accuracy and robustness of finite element
solutions~\cite{CoReBaHu05,EvBaBaHu09,LiEvBaElHu10}.

An important development in the evolution of isogeometric analysis was
the advent of Analysis-suitable T-splines (ASTS). ASTS are a mildly restricted subset of
T-splines which are optimized to simultaneously meet the needs of design and
analysis~\cite{LiZhSeHuSc10,ScLiSeHu10}. Linear independence of
analysis-suitable T-spline blending functions was established
in~\cite{LiZhSeHuSc10}. An efficient local refinement algorithm for
ASTS was developed in~\cite{ScLiSeHu10}. Later, it was shown that a
dual basis, constructed as in the tensor product settting, could be
generalized to ASTS~\cite{BeBuChSa12}. This characteristic of ASTS is
called dual compatibility. These results were then generalized to ASTS
surfaces of arbitrary degree in~\cite{BeBuSaVa12}.

In this paper we continue to develop the theory of ASTS
spaces. Specifically, we provide a rigorous characterization of ASTS
and show that the space of smooth parametric bicubic polynomials,
defined over the extended T-mesh of an ASTS, is contained in the corresponding
ASTS space. To accomplish this, the theory of perturbed ASTS spaces
is developed and an ASTS dimension formula is established in terms of the
topology of the extended T-mesh. We note that, unlike existing
approaches, our dimension formula does not require
that the T-mesh have any particular nesting structure. We then show that this
characterization, coupled with the dual compatibility of ASTS, can be
used to prove that ASTS spaces possess the same optimal approximation
properties as tensor product  B-spline
spaces~\cite{BaBeCoHuSa06}. Next, we prove under what conditions two
ASTS spaces are nested. This provides the theoretical justification for the
analysis-suitable local refinement algorithm in~\cite{ScLiSeHu10} and
provides a foundation upon which adaptive isogeometric analysis
procedures may be developed in the future.

This paper is organized as follows.
Section~\ref{sec:tmesh} describes the T-mesh in index space,
T-junctions, and the extended T-mesh. T-splines in the parametric
domain, blending functions, and T-spline spaces are defined
in Section~\ref{sec:ts_spaces}. Section~\ref{sec:asts} describes the
conditions under which a T-spline is analysis-suitable.
The theory of smoothly perturbed ASTS is developed in
Section~\ref{sec:perturb}. Section~\ref{sec:refine} proves the
conditions under which two ASTS spaces are nested. Using the
characterization of ASTS spaces and dual compatibility several basic
approximation results are proven in Section~\ref{sec:approx}. Finally,
Section~\ref{sec:lemma} proves the dimension of ASTS spaces.

\section{The T-mesh}
\label{sec:tmesh}
An important object underlying T-spline spaces is the T-mesh. A T-mesh
is used to determine T-spline basis functions and how they are
arranged with respect to one another. In other words, the mesh
topology of the T-mesh determines the functional properties of the
resulting space. In an attempt to adhere to a single notation and to
reduce confusion, we define a T-mesh following much of the notation given
in~\cite{LiZhSeHuSc10,BeBuChSa12}. For quick reference, \ref{sec:notation} lists the most
important notational conventions used throughout the text and where
they are defined.

\subsection{Definition}
\label{sec:tmesh_def}
A T-mesh
$\mathsf{T}$ is a rectangular partition of the index domain
$[\underline{m},\overline{m}] \times [\underline{n},\overline{n}]$,
$\underline{m}, \overline{m}, \underline{n}, \overline{n} \in
\mathbb{Z}$, where all rectangle corners (or vertices) have integer
coordinates and all rectangles are open sets.
Each vertex in $\mathsf{T}$ is a singleton subset of
$\mathbb{Z}^2$. We denote all vertices of $\mathsf{T}$ by
$\mathsf{V}$. An edge of $\mathsf{T}$ is a
segment between vertices of $\mathsf{T}$ that does not intersect any
rectangle of $\mathsf{T}$. We note that edges do not
contain vertices and they are open at their
endpoints. We denote all edges of $\mathsf{T}$ by $\mathsf{E}$.
Figure~\ref{fig:tmesh} shows an example of a T-mesh.
The notation $\mathsf{T}^1 \subseteq \mathsf{T}^2$ will indicate
that $\mathsf{T}^2$ can be created by adding vertices and edges to
$\mathsf{T}^1$.

\begin{figure}[htb]
\centering
~\\[-1ex]
\subfigure {\includegraphics [width=3.5in]{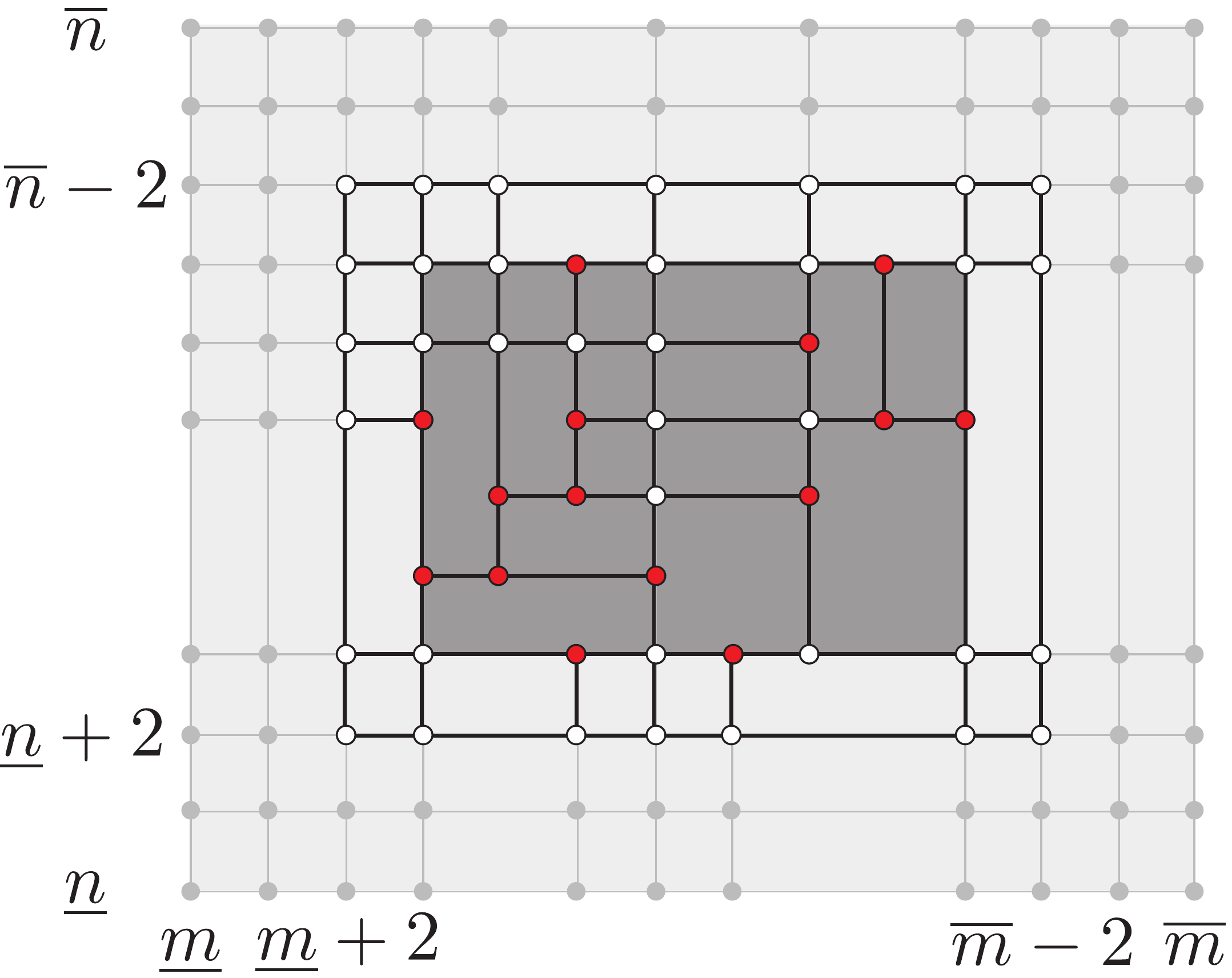}}
~\\[-1ex]
\caption{A T-mesh.} \label{fig:tmesh}
~\\[-1ex]
\end{figure}

The valence of a vertex $V \in \mathsf{V}$ is
the number of edges such that $V$ is an endpoint. We only allow valence
three (called T-junctions) or four vertices. Note that valence two vertices,
other than the four corners, are eliminated
from the definition.

The horizontal (resp., vertical) skeleton of a T-mesh is denoted by
$h\mathsf{S}$ (resp., $v\mathsf{S}$), and is the union of all
horizontal (resp., vertical) edges and all vertices. Finally, we
denote the skeleton to be the union $\mathsf{S} = h\mathsf{S} \cup
v\mathsf{S}$. For a given vertex $a = \{(i,j)\}$
we define $h\mathsf{J}(a) := \{k \in \mathbb{Z} : \{k\} \times a
\subset v\mathsf{S} \}$ and $v\mathsf{J}(a) := \{k \in \mathbb{Z} : a
\times \{k\} \subset h\mathsf{S} \}$. We assume that these two sets
are ordered.

We split the index domain $\mathsf{R} = [\underline{m}, \overline{m}] \times
[\underline{n}, \overline{n}]$ into an active region $\mathsf{AR}$ and a
frame region $\mathsf{FR}$ such that $\mathsf{R} = \mathsf{FR} \cup
\mathsf{AR}$ and $\mathsf{AR} = [\underline{m} + 2,
\overline{m} - 2] \times [\underline{n} + 2, \overline{n} - 2]$, and
$\mathsf{FR} = ([\underline{m}, \underline{m} + 2] \cup [\overline{m} - 2,
 \overline{m}]) \times [\underline{n}, \overline{n}] \cup [\underline{m}, \overline{m}]
 \times ([\underline{n}, \underline{n} + 2] \cup [\overline{n} - 2,
 \overline{n}])$. Note that both $\mathsf{FR}$ and $\mathsf{AR}$ are closed.

A symbolic T-mesh~\cite{LiZhSeHuSc10} is created from a T-mesh $\mathsf{T}$
by assigning a symbol in Table~\ref{tab:symbol} to each
vertex in a tensor product mesh formed from the index
coordinates, $\{\underline{m}, \ldots,
\overline{m}\} \times \{\underline{n}, \ldots, \overline{n}\}
\subset \mathbb{Z}^2$. The
symbol is chosen to match the mesh topology of $\mathsf{T}$. The
symbolic T-mesh corresponding to the T-mesh in Figure~\ref{fig:tmesh}
is shown in Figure~\ref{fig:tmesh_symbol}.

\begin{table}[htbp]
  \centering
  \caption{Definition of possible symbols in a symbolic T-mesh
    \label{tab:symbol}}
  \begin{tabular}{c c}
    \hline
    Symbol & Correspondence with $\mathsf{T}$ \\ \hline
    $+$ & Valence 4 vertex, corner vertex, or valence 3
    boundary vertex in $\mathsf{T}$ \\
    $\vdash$, $\dashv$, $\bot$, $\top$ & Oriented valence three vertex
    in $\mathsf{T}$\\
    $|$ & Vertical edge in $\mathsf{T}$ \\
    $-$ & Horizontal edge in $\mathsf{T}$ \\
    $\cdot$ & No corresponding vertex or edge in $\mathsf{T}$ \\
    \hline
  \end{tabular}
\end{table}

\begin{figure}[htb]
\centering
~\\[-1ex]
\subfigure {\includegraphics [width=3.5in]{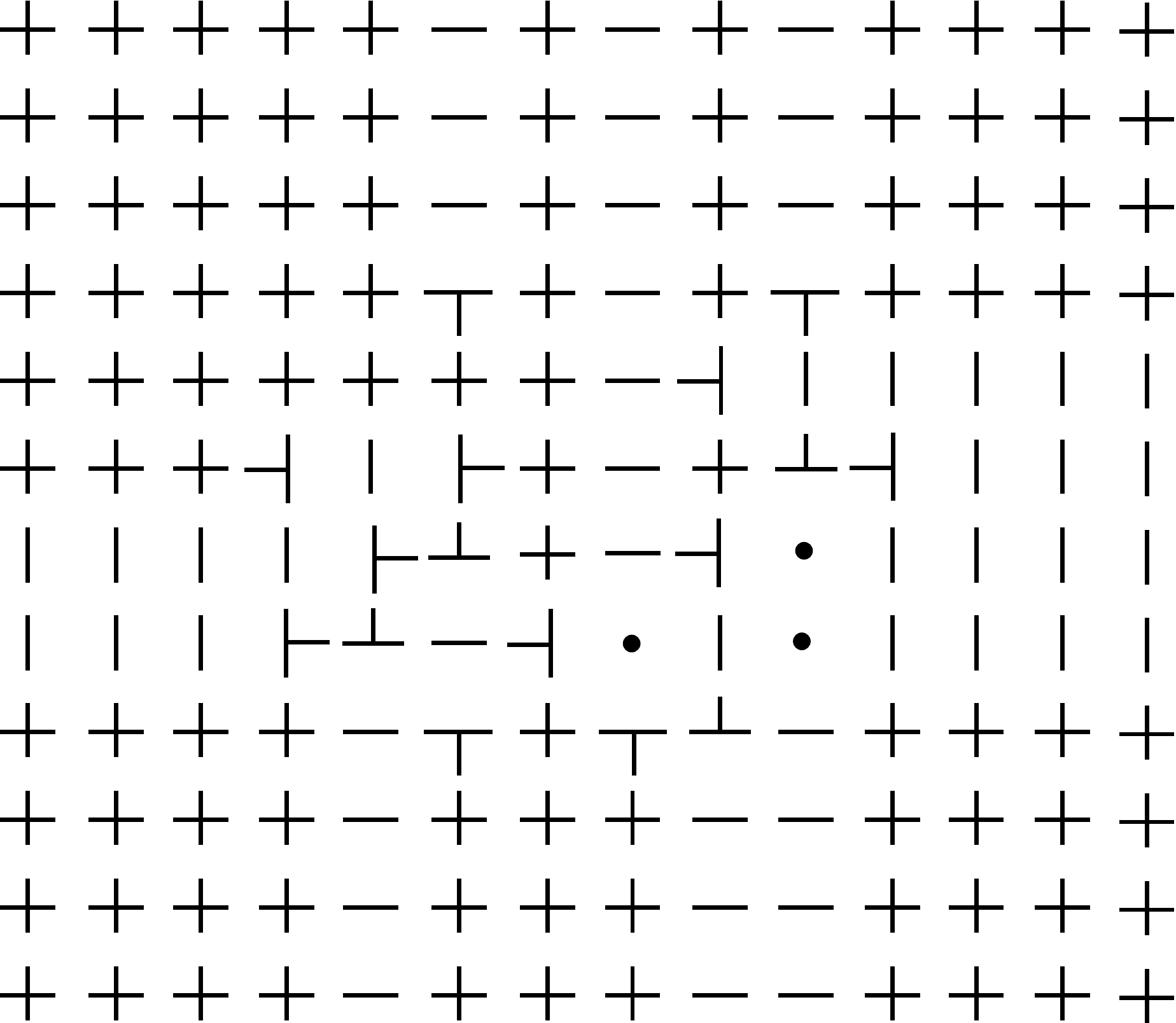}}
~\\[-1ex]
\caption{The symbolic representation of the T-mesh in
  Figure~\ref{fig:tmesh}.}
\label{fig:tmesh_symbol}
~\\[-1ex]
\end{figure}

\subsection{Admissible T-meshes}
We say that a T-mesh is admissible if it satisfies three basic
conditions. First, we require that $\mathsf{S} \cap \mathsf{FR}$
contains the vertical segments $\{i\} \times [\underline{n},
\overline{n}]$ for $i = \underline{m}, \underline{m} + 1,
\underline{m} + 2, \overline{m} - 2, \overline{m} - 1, \overline{m}$
and the horizontal segments $[\underline{m}, \overline{m}] \times
\{j\}$ for $j = \underline{n}, \underline{n} + 1,
\underline{n} + 2, \overline{n} - 2, \overline{n} - 1,
\overline{n}$. These horizontal and vertical lines are for basis function definition
near the boundary. Second, we require that $\mathsf{S} \cap
\mathsf{AR}$ contains the vertical segments $\{i\} \times [\underline{n},
\overline{n}]$ for $i = \underline{m} + 2, \underline{m} + 3,
\overline{m} - 3, \overline{m} - 2$
and the horizontal segments $[\underline{m}, \overline{m}] \times
\{j\}$ for $j = \underline{n} + 2, \underline{n} + 3,
\overline{n} - 3, \overline{n} - 2$. Third, we require that for any
two vertices $V_1 = \{(i_1, j_1)\}, V_2=\{(i_2, j_2)\}$ in
$\mathsf{V}$, such that $V_1, V_2 \subset \partial Q$ for some $Q \in
\mathsf{T}$, if $i_1 = i_2$ (resp., $j_1 = j_2$), then $\{i_1\} \times
]j_1, j_2[ \subset \mathsf{S}$ (resp., $]i_1, i_2[ \times \{j_1\} \subset
\mathsf{S}$). From a practical point of view these are minor
restrictions. The T-mesh in Figure~\ref{fig:tmesh} is admissible.

We note that for convenience and simplicity, we often refer to only
the active region of an admissible T-mesh when speaking of a T-mesh. In all cases, we
assume that the frame region has an admissible topology.

\subsection{Anchors and T-junctions}
\label{sec:anchors}
We define the anchors $\mathsf{A}(\mathsf{T}) = \{A \in \mathsf{V} \cap \mathsf{AR}\}$.
We denote the total number of anchors in $\mathsf{T}$ by
$n^A$. We define $\mathsf{J} \subset \mathsf{A}(\mathsf{T})$ to be the set of
all valence three vertices. These are called T-junctions. The symbols
$\vdash$, $\dashv$, $\bot$, $\top$ indicate
the four possible orientations of a T-junction in a symbolic T-mesh. A
T-junction $T_h \in \mathsf{J}$ (resp., $T_v \in \mathsf{J}$) of type
$\vdash$ and $\dashv$ (resp., $\bot$, $\top$) and their
extensions are called horizontal (resp., vertical). The solid white
and red circles in Figure~\ref{fig:tmesh} are anchors and the red
circles are T-junctions.

\subsection{Segments}
\label{sec:segments}
We define a segment to be a closed line segment of contiguous vertices and edges
whose beginning and ending vertices are T-junctions (interior or
boundary). Given two horizontal (resp., vertical) segments $G_1^h, G_2^h$
defined over the intervals $[i_1, j_1] \times a$ and $[i_2, j_2]
\times b$ we say that $G_1^h \leq G_2^h$ if $i_1 \leq i_2$. We denote by
$h\mathsf{G}$ (resp., $v\mathsf{G}$) the
collection of all horizontal (resp., vertical) segments, and by
$\mathsf{G} = h\mathsf{G} \cup v\mathsf{G}$ the collection of all
segments. We define $h\mathsf{G}(a) = h\mathsf{G} \cap ([\underline{m},
\overline{m}] \times a)$ and $v\mathsf{G}(a) = v\mathsf{G} \cap (a \times [\underline{n},
\overline{n}])$. We assume these two sets are ordered. We denote the total
number of segments in $\mathsf{T}$ by $n^{G}$. We denote the total
number of horizontal (resp., vertical) segments in $\mathsf{T}$ by
$n^{G}_h$ (resp., $n^{G}_v$). We denote the number of line segments
in $h\mathsf{G}(a)$ (resp., $v\mathsf{G}(a)$) by $n^{G}_h(a)$ (resp.,
$n^{G}_v(a)$).

\subsection{The extended T-mesh}
\label{sec:ext_tmesh}
T-junction extensions can be associated with each T-junction. For
example, given a T-junction $T = \{(\overline{\imath},\overline{\jmath})\} \in \mathsf{J}$ of
type $\vdash$ we extract from $h\mathsf{J}(\{\overline{\jmath}\})$
four consecutive indices $i_1, \ldots, i_4$ such that $\overline{\imath} =
i_3$. We call $ext^e(T)=[i_1, \overline{\imath}] \times
\{\overline{\jmath}\}$ the face
extension, $ext^f(T) = ]\overline{\imath}, i_4] \times
\{\overline{\jmath}\}$ the edge
extension for such kind of T-junction. Similarly, we can define the face and edge extensions for the other kinds
of T-junctions $\dashv$, $\bot$, $\top$ which are illustrated in Figure~\ref{fig:extension}.

\begin{figure}[htb]
\centering
~\\[-1ex]
\subfigure {\includegraphics [width=4.5in]{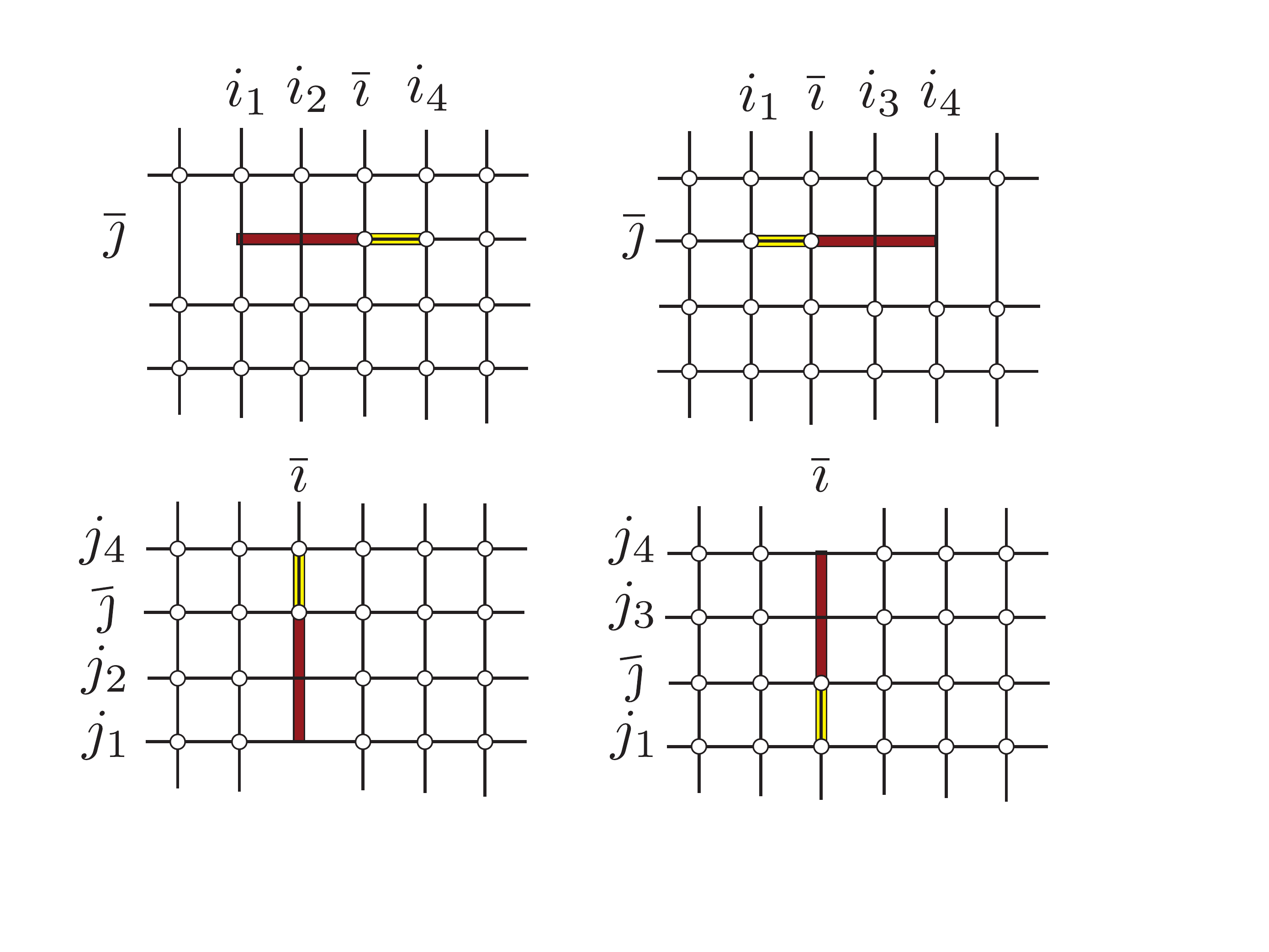}}
~\\[-1ex]
\caption{The face extension (red) and edge extension (yellow) for four different kinds of
T-junctions.\label{fig:extension}}
~\\[-1ex]
\end{figure}

We denote $ext(T) = ext^e(\mathsf{T}) \cup ext^f(T)$ the extension of
T-junction $T$ and the union of all horizontal (resp., vertical) face
extensions by $hext^f(\mathsf{T})$ (resp., $vext^f(\mathsf{T})$), the
union of all face extensions by $ext^f(\mathsf{T})$, and
the union of all extensions (face and edge) by $ext(\mathsf{T})$.
We define the extended T-mesh, $\mathsf{T}_{ext}$, as the T-mesh
created by adding to $\mathsf{T}$ all the T-junction extensions. In
other words, $\mathsf{T}_{ext} = \mathsf{T} \cup ext(\mathsf{T})$. We
denote the total number of vertices in $\mathsf{T}_{ext}$ by
$n^{ext}$. The extended T-mesh corresponding to the T-mesh in
Figure~\ref{fig:tmesh} is shown in Figure~\ref{fig:tmesh_extended}.

\begin{figure}[htb]
\centering
~\\[-1ex]
\subfigure {\includegraphics [width=3.5in]{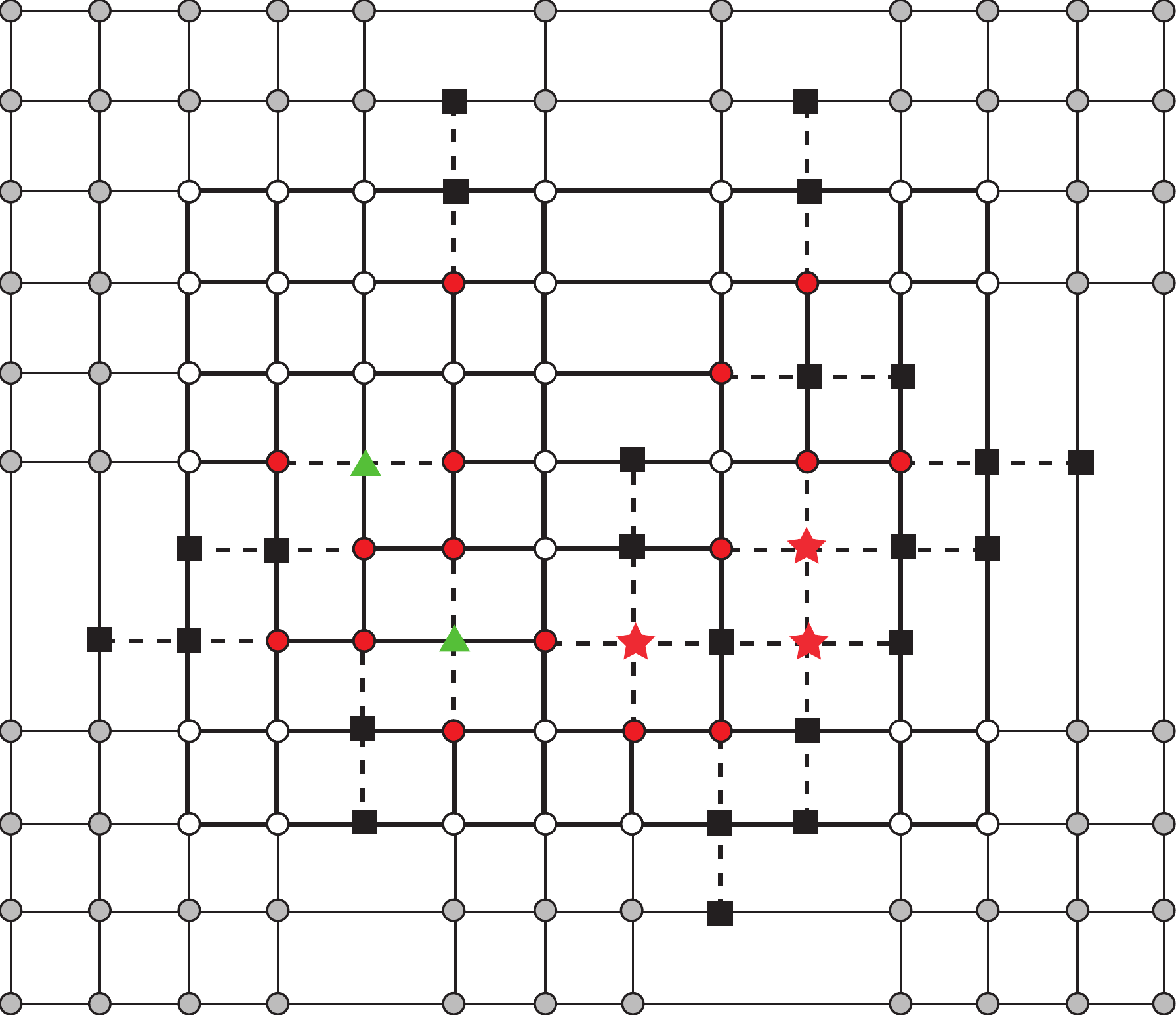}}
~\\[-1ex]
\caption{The extended T-mesh corresponding to the T-mesh in
  Figure~\ref{fig:tmesh}.  The crossing vertices are denoted by red stars, the overlap
vertices are denoted by green triangles, and the extended
vertices are denoted by black squares. Notice that the inactive vertices (grey circles)
are also regarded as extended vertices. The active vertices are denoted by hollow
circles and the T-junctions are denoted by red circles.}
\label{fig:tmesh_extended}
~\\[-1ex]
\end{figure}

Adding T-junction extensions to $\mathsf{T}$ may introduce three
additional collections of vertices. The first, called crossing
vertices and denoted by $\mathsf{CV}$, is
created from the intersection of crossing face extensions. In other
words, $$\mathsf{CV} = hext(\mathsf{T}) \cap vext(\mathsf{T}).$$ We
denote the number of crossing vertices in $\mathsf{T}_{ext}$ by
$n^+$. In Figure~\ref{fig:tmesh_extended} the crossing vertices are
denoted by red stars.

The second, called overlap vertices and denoted by $\mathsf{OV}$, is
created from the intersection of overlapping face extensions with
$\mathsf{S}$. In other words, $$\mathsf{OV} = ((\bigcap_{T_h \in
  \mathsf{J}} ext^f(T_h)) \cap v\mathsf{S}) \cup  ((\bigcap_{T_v \in
  \mathsf{J}} ext^f(T_v)) \cap h\mathsf{S}).$$ We denote the number of
overlap vertices in $\mathsf{T}_{ext}$ by $n^-$.  In
Figure~\ref{fig:tmesh_extended} the overlap vertices are
denoted by green triangles.

The third, called
extended vertices and denoted by $\mathsf{EV}$, is created from the
intersection of face extensions and $\mathsf{S}$ while removing those
vertices which already correspond to overlap vertices. Additionally,
all non-anchor vertices are classified as extended vertices. In other
words, $$\mathsf{EV} = ((ext^f(\mathsf{T}) \cap \mathsf{S}) \setminus
\mathsf{OV}) \cup (\mathsf{V} \setminus \mathsf{A}(\mathsf{T})).$$ We
denote the number of extended vertices in $\mathsf{T}_{ext}$ by
$n^*$.  In Figure~\ref{fig:tmesh_extended} the extended vertices are
denoted by black squares.

\section{The parametric domain and T-spline spaces}
\label{sec:ts_spaces}
Let $\Xi = (\xi_{\underline{m}}, \ldots,
\xi_{\overline{m}})$ and $\Pi = (\eta_{\underline{n}}, \ldots, \eta_{\overline{n}})$
be two global knot vectors defined on the interval $\mathbb{R}$. Interior
knots may have a multiplicity of three while end knots may have a
multiplicity of four. The global knot vectors
define a full parametric domain, $\tilde{\Omega} \subset
\mathbb{R}^2$, where $\tilde{\Omega} = [\xi_{\underline{m}}, \xi_{\overline{m}}] \otimes
[\eta_{\underline{n}}, \eta_{\overline{n}}]$ and a reduced parametric domain,
$\hat{\Omega} \subset \tilde{\Omega}$, where $\hat{\Omega} =
[\xi_{\underline{m}+3}, \xi_{\overline{m}-3}] \otimes
[\eta_{\underline{n}+3}, \eta_{\overline{n}-3}]$. The T-mesh in the
parametric domain is defined as the collection of
\textit{non-empty} elements of
the form $\tilde{Q} = ]\xi_{i_1}, \xi_{i_2}[ \times ]\eta_{j_1},
\eta_{j_2}[$ where $Q=]i_1, i_2[ \times ]j_1, j_2[ \in
\mathsf{T}$. We denote those elements where $\tilde{Q} \cap
\hat{\Omega} \neq \emptyset$ by $\hat{Q}$. The extended T-mesh in the
parametric domain as well as all element related concepts
are defined similarly. Throughout this paper we use the index and
parametric representation of a T-mesh interchangeably with the context
making the use clear.

For each anchor $A = a \times b \in
\mathsf{A}(\mathsf{T})$ we define its horizontal (vertical) index
vector $hv(A)$ ($vv(A)$, respectively) as a subset of $h\mathsf{J}(b)$
($v\mathsf{J}(a)$, respectively) where $hv(A) = (i_1, \ldots, i_5) \in
\mathbb{Z}^5$ contains five unique consecutive indices in
$h\mathsf{J}(b)$ with $\{i_3\} = a$. The vertical index vector,
denoted by $vv(A)$, is constructed in an analogous manner. We then
associate a T-spline blending function $N_A(\xi, \eta)$ with anchor
$A$. The T-spline blending functions are given by
\begin{equation}
N_A(\xi, \eta) := B[\Xi_A](\xi) B[\Pi_A](\eta) \quad \forall (\xi,
\eta) \in \hat{\Omega}
\end{equation}
where $B[\Xi_A](\xi)$ and $B[\Pi_A](\eta)$ are the cubic B-spline
basis functions associated with the local knot vectors
\begin{align}
\Xi_A &= [\xi_{i_1}, \ldots, \xi_{i_5}] \subset \Xi \\
\Pi_A &= [\eta_{j_1}, \ldots, \eta_{j_5}] \subset \Pi
\end{align}
and $hv(A) = (i_1, \ldots, i_5)$ and $vv(A) = (j_1, \ldots, j_5)$.

Figure~\ref{fig:bf} illustrates the construction of a T-spline
blending function corresponding to anchor $A=\{(3,3)\}$. In
this case, the local knot vectors are $\Xi_A = [\xi_1, \xi_2, \xi_3,
\xi_4, \xi_6]$ and $\Pi_A = [\eta_1, \eta_2, \eta_3, \eta_4,
\eta_5]$.

\begin{figure}[htb]
\centering
~\\[-1ex]
\subfigure {\includegraphics [width=3.5in]{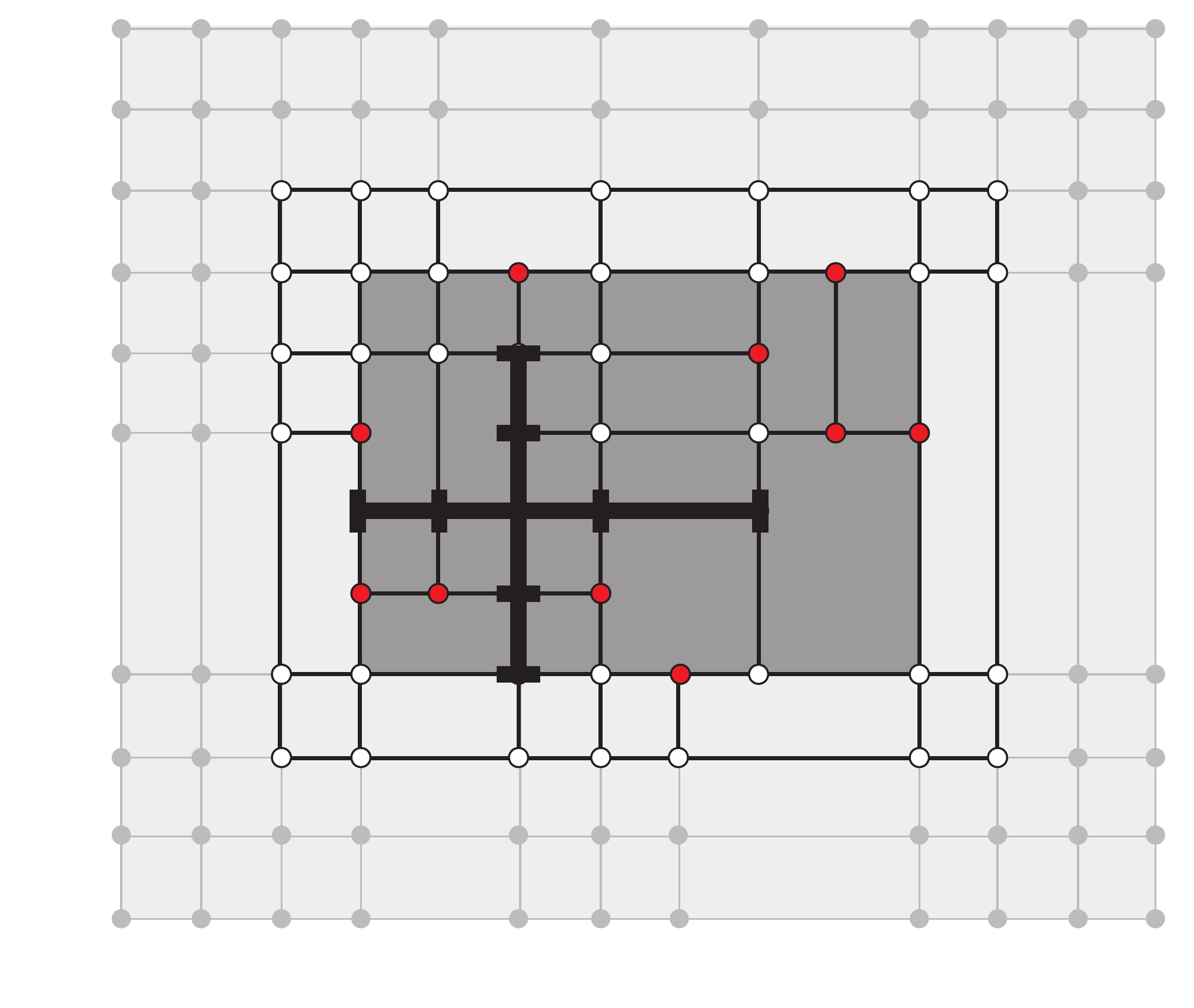}}
~\\[-1ex]
\caption{Inferring a T-spline blending function from a T-mesh.} \label{fig:bf}
~\\[-1ex]
\end{figure}

A T-spline space $\mathcal{T} = \mathcal{T}(\mathsf{T}, \Xi, \Pi)$ is
simply the span of the blending functions, $N_A$, $A \in
\mathsf{A}(\mathsf{T})$.

\section{Analysis-suitable T-splines}
\label{sec:asts}
Analysis-suitable T-splines form a practically useful subset of T-splines.
ASTS maintain the important mathematical
properties of the NURBS basis while providing an efficient and highly
localized refinement capability. Several important properties of
ASTS have been proven:
\begin{itemize}
\item The blending functions are linearly
  independent for \textit{any} choice of knots~\cite{LiZhSeHuSc10}.
\item The basis constitutes a partition of unity (see Corollary~\ref{cor:pu}).
\item Each basis function is non-negative.
\item They can be generalized to arbitrary degree~\cite{BeBuSaVa12}.
\item An affine transformation of an analysis-suitable T-spline is
  obtained by applying the transformation to the control points. We
  refer to this as affine covariance. This implies that all ``patch
  tests'' (see~\cite{Hug00}) are satisfied \textit{a priori}.
\item They obey the convex hull property.
\item They can be locally refined~\cite{SeZhBaNa03,ScLiSeHu10}.
\item A dual basis can be constructed~\cite{BeBuChSa12,BeBuSaVa12}.
\end{itemize}

\begin{definition}
An analysis-suitable T-spline is a T-spline whose T-mesh is
analysis-suitable~\cite{LiZhSeHuSc10}. A T-mesh is said to be
analysis-suitable if it is admissible and no horizontal T-junction extension
intersects a vertical T-junction extension.
\end{definition}

An analysis-suitable T-mesh is shown in
Figure~\ref{fig:as_tmesh}a. The corresponding extended T-mesh is shown in
Figure~\ref{fig:as_tmesh}b. Notice that no horizontal extension
intersects a vertical extension. The dual basis for an ASTS equips these spaces with a rich
mathematical structure which we leverage in this
paper~\cite{BeBuChSa12}.
\begin{lemma}
For a bicubic ASTS, each dual basis function, corresponding to a
T-spline basis function $N_A(\xi, \eta)$ with local knot vectors $\Xi_A =
[\xi_{i_1}, \ldots, \xi_{i_5}]$ and $\Pi_A = [\eta_{j_1}, \ldots,
\eta_{j_5}]$, is
$$\lambda_{A} = \lambda[\xi_{i_1}, \ldots, \xi_{i_5}]\otimes
\lambda[\eta_{j_1}, \ldots, \eta_{j_5}]$$
where $\lambda[\xi_{i_1}, \ldots, \xi_{i_5}]$ and $\lambda[\eta_{j_1},
\ldots, \eta_{j_5}]$ are dual basis functions corresponding to
univariate cubic B-splines~\cite{Schu93} whose knot vectors are
$[\xi_{i_1}, \ldots, \xi_{i_5}]$ and $[\eta_{j_1},
\ldots, \eta_{j_5}]$, respectively.
\end{lemma}

\begin{figure}[htb]
\centering
~\\[-1ex]
\subfigure {\includegraphics [width=4in]{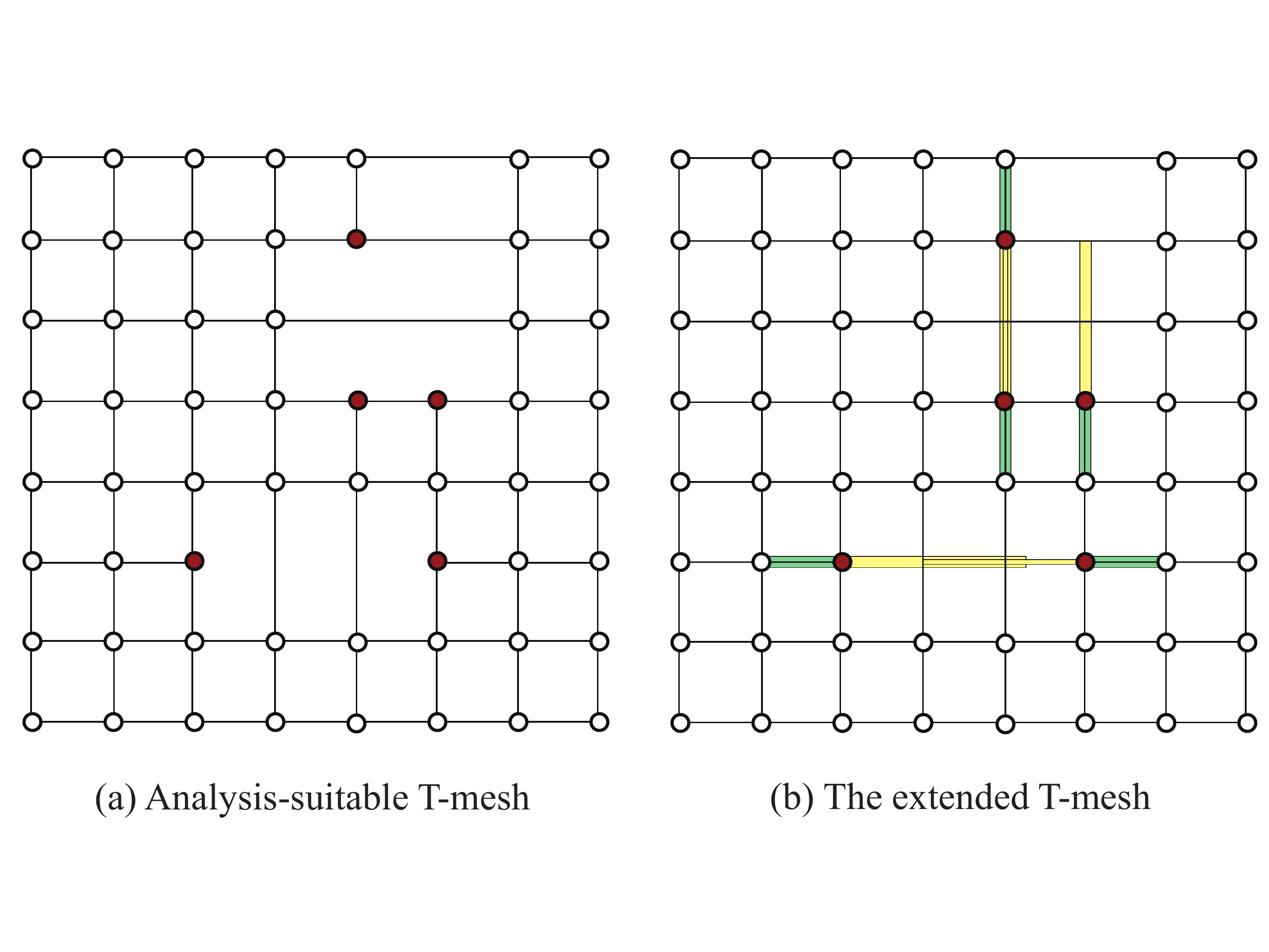}}
~\\[-1ex]
\caption{An analysis-suitable T-mesh (a) and the extended T-mesh (b).}
\label{fig:as_tmesh}
~\\[-1ex]
\end{figure}

\section{Perturbed T-splines}
\label{sec:perturb}
From a theoretical point of view, developing a complete and rigorous
characterization of T-spline spaces is
complicated by the presence of zero knot intervals (especially near
T-junctions) and overlap vertices.  However, allowing both is important when
T-splines are used as a tool in design and analysis.

To overcome this difficulty, we develop the theory of the perturbed T-mesh (and
resulting perturbed T-spline space). A perturbed T-spline can be used to prove
properties about the \textit{original} T-spline. In other words, we
will generate a perturbed T-mesh, establish the result in the
perturbed setting, and then show that the result holds as the
perturbation converges to the original T-spline.

\subsection{Perturbed T-meshes}
\label{subsec:perturb}
A perturbed T-mesh is created by first generating perturbed global knot
vectors, $\Xi[\boldsymbol{\delta}]$, $\Pi[\boldsymbol{\delta}]$, where
$\boldsymbol{\delta} =(\delta_1, \ldots, \delta_n)^T$ is a vector of
perturbation parameters. A perturbed global knot vector is
written as $$\Xi[\boldsymbol{\delta}] = (\xi[\boldsymbol{\delta}]_{\underline{\imath}}, \ldots,
\xi[\boldsymbol{\delta}]_{\imath}, \ldots,
\xi[\boldsymbol{\delta}]_{\overline{\imath}})$$ where $\imath = \imath(i,g)$ takes
the index of the $i^{th}$ knot in $\Xi$ and the $g^{th}$ segment in
$v\mathsf{G}(\{i\})$ and returns a unique index in the perturbed global
knot vector. The knot values are initialized as $\xi[\boldsymbol{\delta}]_{\imath(i,g)}
= \xi_i$. In other words, a knot index which corresponds to a
$v\mathsf{G}(\{i\})$ which contains multiple
segments in the T-mesh is repeated $n_v^G(\{i\})$ times. Notice that
this operation induces an index map $h\pi(\imath(i,g)) = i$ (resp., $v\pi$) from the
indices in the perturbed global knot vector onto the original global knot
vector. The knot values are then perturbed using a small parameter
$\delta \in \mathbb{R}$ as
\begin{equation*}
\xi[\boldsymbol{\delta}]_{\imath} = \xi[\boldsymbol{\delta}]_{\underline{\imath}} + \sum_{\jmath=\underline{\imath} +1}^{\imath}
\Delta \xi[\boldsymbol{\delta}]_{\jmath}, \quad \imath = \underline{\imath}, \ldots, \overline{\imath}
\end{equation*}
where $\Delta \xi[\boldsymbol{\delta}]_{\jmath} = c_{\alpha(\jmath)}\delta = \delta_{\alpha}$, if $\xi_{h\pi(\jmath)} - \xi_{h\pi({\jmath}-1)} =
0$, and is equal to $\xi_{h\pi(\jmath)} - \xi_{h\pi({\jmath}-1)}$,
otherwise. The constant, $c_{\alpha} \in [0, \infty)$. This same
procedure is applied to $\Pi$ to form
$\Pi[\boldsymbol{\delta}]$. The T-mesh, $\mathsf{T}$, is then modified to form the
perturbed T-mesh, $\mathsf{T}[\boldsymbol{\delta}]$, by associating the
vertices and edges contained in the $g^{th}$ segment of
$v\mathsf{G}(\{i\})$ with knot $\xi[\boldsymbol{\delta}]_{\imath(i, g)}$. Notice
that the number of anchors does not change when forming
$\mathsf{T}[\boldsymbol{\delta}]$. A perturbed T-spline space $\mathcal{T}[\boldsymbol{\delta}]$
is a T-spline space formed from perturbed global knot vectors and
T-mesh. A \textit{strictly} perturbed T-mesh or T-spline space
is one where $c_{\alpha} \in (0, \infty)$.

A perturbation of an analysis-suitable T-mesh is shown in
Figure~\ref{fig:perturb}. The analysis-suitable T-mesh is shown in
Figure~\ref{fig:perturb}a and the perturbed T-mesh is shown in
Figure~\ref{fig:perturb}b. Knot intervals
are shown instead of knots for simplicity. Recall that a knot interval
is simply the difference between adjacent knots in a global knot vector. Notice that the
horizontal and vertical zero knot intervals have been replaced by
non-zero knot intervals $\sigma_2$ and $\sigma_3$. The vertical
segments with T-junctions are perturbed resulting in a new knot
interval $\sigma_1$.

\begin{figure}[htb]
\centering
~\\[-1ex]
\subfigure {\includegraphics [width=4.5in]{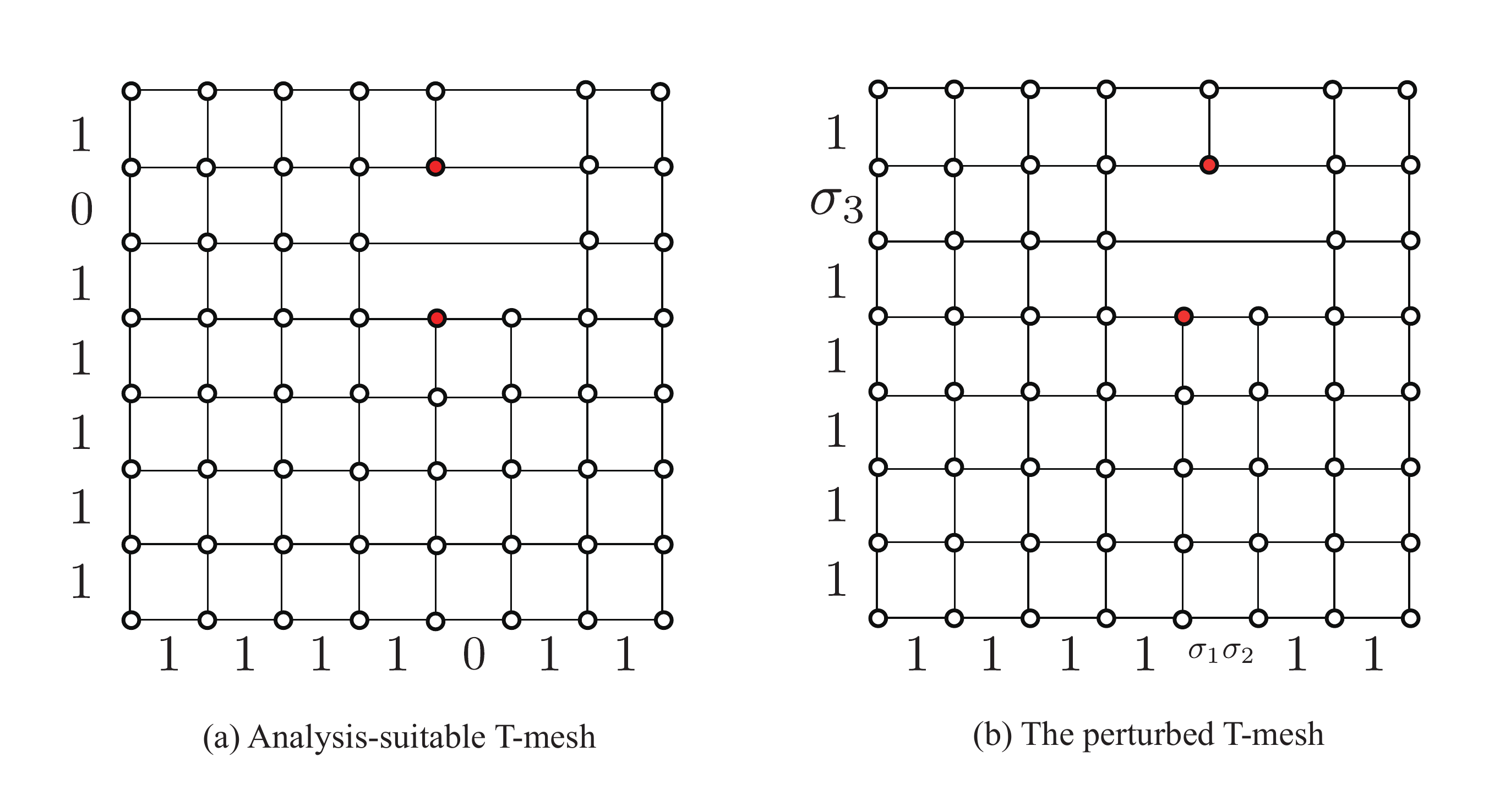}}
~\\[-1ex]
\caption{A perturbation of an analysis-suitable T-mesh $\mathsf{T}$
  shown in (a) results in the perturbed T-mesh
  $\mathsf{T}[\boldsymbol{\delta}]$ shown in (b).}
\label{fig:perturb}
~\\[-1ex]
\end{figure}

\begin{proposition}
If $\mathsf{T}$ is analysis-suitable then $\mathsf{T}[\boldsymbol{\delta}]$ is
analysis-suitable.
\label{lemma:offset_asts}
\end{proposition}
\begin{proof}
Suppose the extensions of two T-junctions
$ext(V_{1}) = [i_{1}^{1}, i_{4}^{1}] \times \{j^{1}\}$ ($\vdash$ or $\dashv$)
and $ext(V_{2}) = \{i^{2}\} \times [j_{1}^{2}, j_{4}^{2}]$ ($\bot$ or
$\top$) in $\mathsf{T}[\boldsymbol{\delta}]$
intersect. This implies that $j_{1}^{2} \leq j^{1} \leq j_{4}^{2}$ and
$i_{1}^{1} \leq i^{2} \leq i_{4}^{1}$. According to the construction
of $\mathsf{T}[\boldsymbol{\delta}]$, we don't change the order
of the indices, so $h\pi(j_{1}^{2}) \leq h\pi(j^{1}) \leq h\pi(j_{4}^{2})$ and
$v\pi(i_{1}^{1}) \leq v\pi(i^{2}) \leq v\pi(i_{4}^{1})$, i.e., there are
intersecting extensions in $\mathsf{T}$.
\end{proof}

\begin{lemma}
Let $\mathsf{T}$ be an analysis-suitable T-mesh. For every anchor,
$\{(\imath, \jmath)\}$, and horizontal index vector,
$hv(\{(\imath, \jmath)\})[\boldsymbol{\delta}]$, in the perturbed
T-mesh, $\mathsf{T}[\boldsymbol{\delta}]$, $hv(\{(h\pi(\imath),v\pi(\jmath))\}) =
h\pi(hv(\{(\imath, \jmath)\})[\boldsymbol{\delta}])$ where
$h\pi(hv(\{(\imath, \jmath)\})[\boldsymbol{\delta}]) = (h\pi(\imath_1),
\ldots, h\pi(\imath_5) )$. This also holds for the vertical
index vectors.
\label{lemma:map_knots}
\end{lemma}
\begin{proof}
If the topological symbols corresponding to the indices $\{(i_1,
\jmath), \ldots, (i_5, \jmath)\}$ where $i_\ell \in hv(\{(\imath,
\jmath)\})[\boldsymbol{\delta}]$, $\ell = 1,\ldots,5$, are not $\vdash$ or
$\dashv$ then the result immediately follows since the indices are
contained in a single horizontal segment. Thus, we only need to prove
that the result holds when the topological symbol corresponding to
$\{(\imath, \jmath)\}$ is $\vdash$ or
$\dashv$. Without loss of generality we assume it
to be $\vdash$. Since, according to Lemma~\ref{lemma:offset_asts}, the
T-mesh, $\mathsf{T}[\boldsymbol{\delta}]$, is analysis-suitable the
symbols for the first two indices in $hv(\{(\imath,
\jmath)\})[\boldsymbol{\delta}]$ can only be $|$ or $\dashv$. Thus,
$hv(\{(h\pi(\imath),v\pi(\jmath))\}) = h\pi(hv(\{(\imath, \jmath)\})[\boldsymbol{\delta}])$.
\end{proof}

\begin{theorem}
If $\mathcal{T}$ is analysis-suitable and $\mathcal{T}[\boldsymbol{\delta}]$ is an
offset perturbation then $N_A[\boldsymbol{\delta}] \rightarrow N_{\pi(A)}$
as $\delta \rightarrow 0$, $A \in
\mathsf{A}(\mathsf{T}[\boldsymbol{\delta}])$ where $\pi(A) =
\{(h\pi(\imath), v\pi(\jmath))\}$.
\label{thm:converge_basis}
\end{theorem}
\begin{proof}
Notice that if $\delta \rightarrow 0$,
$\xi[\boldsymbol{\delta}]_{\imath} \rightarrow \xi_{h\pi(\imath)}$ for
any $\imath$ and $\eta[\delta]_{\jmath} \rightarrow
\eta_{v\pi(\jmath)}$ for any $\jmath$. According to Lemma~\ref{lemma:map_knots},
the local knot vectors for each basis function $N_A[\boldsymbol{\delta}]$ converge to those for
$N_{\pi(A)}$, $A \in \mathsf{A}(\mathsf{T}[\boldsymbol{\delta}])$. Thus, according
to (\cite{Schu93}, Theorem 4.36) we have the result.
\end{proof}

\section{Refineability and nestedness}
\label{sec:refine}
We now explore the refineability and nesting behavior of
analysis-suitable T-spline spaces. In
other words, given two analysis-suitable T-splines spaces,
$\mathcal{T}^1$ and $\mathcal{T}^2$, we establish the conditions under
which $\mathcal{T}^1 \subseteq \mathcal{T}^2$. We first establish basic refineability
properties when the analysis-suitable T-mesh does not have any knot
multiplicities or overlap vertices. Using the theory of perturbed T-splines, we then
extend those results to encompass T-meshes
which do have zero knot intervals and overlap vertices.

\begin{definition}
The notation $\mathsf{T}^1[\boldsymbol{\delta}, \mathsf{T}^2]$ denotes
a perturbed T-mesh where $\mathsf{T}^1 \subseteq \mathsf{T}^2$ and
$\mathsf{T}^1[\boldsymbol{\delta}, \mathsf{T}^2]$ is created by
removing those edges and vertices from the strictly perturbed T-mesh
$\mathsf{T}^2[\boldsymbol{\delta}]$ which correspond to unperturbed edges and
vertices in $\mathsf{T}^2 \setminus
\mathsf{T}^1$. By inspection, it is clear that $\mathsf{T}^1[\boldsymbol{\delta}, \mathsf{T}^2]$,
constructed in this way, is a perturbed T-mesh which satisfies
Proposition~\ref{lemma:offset_asts},
Lemma~\ref{lemma:map_knots}, and Theorem~\ref{thm:converge_basis} and
that $\mathsf{T}^1[\boldsymbol{\delta}, \mathsf{T}^2]
\subseteq \mathsf{T}^2[\boldsymbol{\delta}]$.
\end{definition}

The construction of $\mathsf{T}^1[\boldsymbol{\delta},
\mathsf{T}^2]$ is depicted in Figure~\ref{fig:submesh}. Two
analysis-suitable T-meshes are shown in Figure~\ref{fig:submesh}a and
Figure~\ref{fig:submesh}b. Notice that $\mathsf{T}^1 \subseteq
\mathsf{T}^2$. The perturbed T-mesh $\mathsf{T}^1[\boldsymbol{\delta},
\mathsf{T}^2]$ (shown in Figure~\ref{fig:submesh}c) is formed by
removing the dotted lines (shown in Figure~\ref{fig:submesh}c) from
$\mathsf{T}^2[\boldsymbol{\delta}]$ (shown in
Figure~\ref{fig:submesh}d).

\begin{figure}[htb]
\centering
\subfigure {\includegraphics [width=4in]{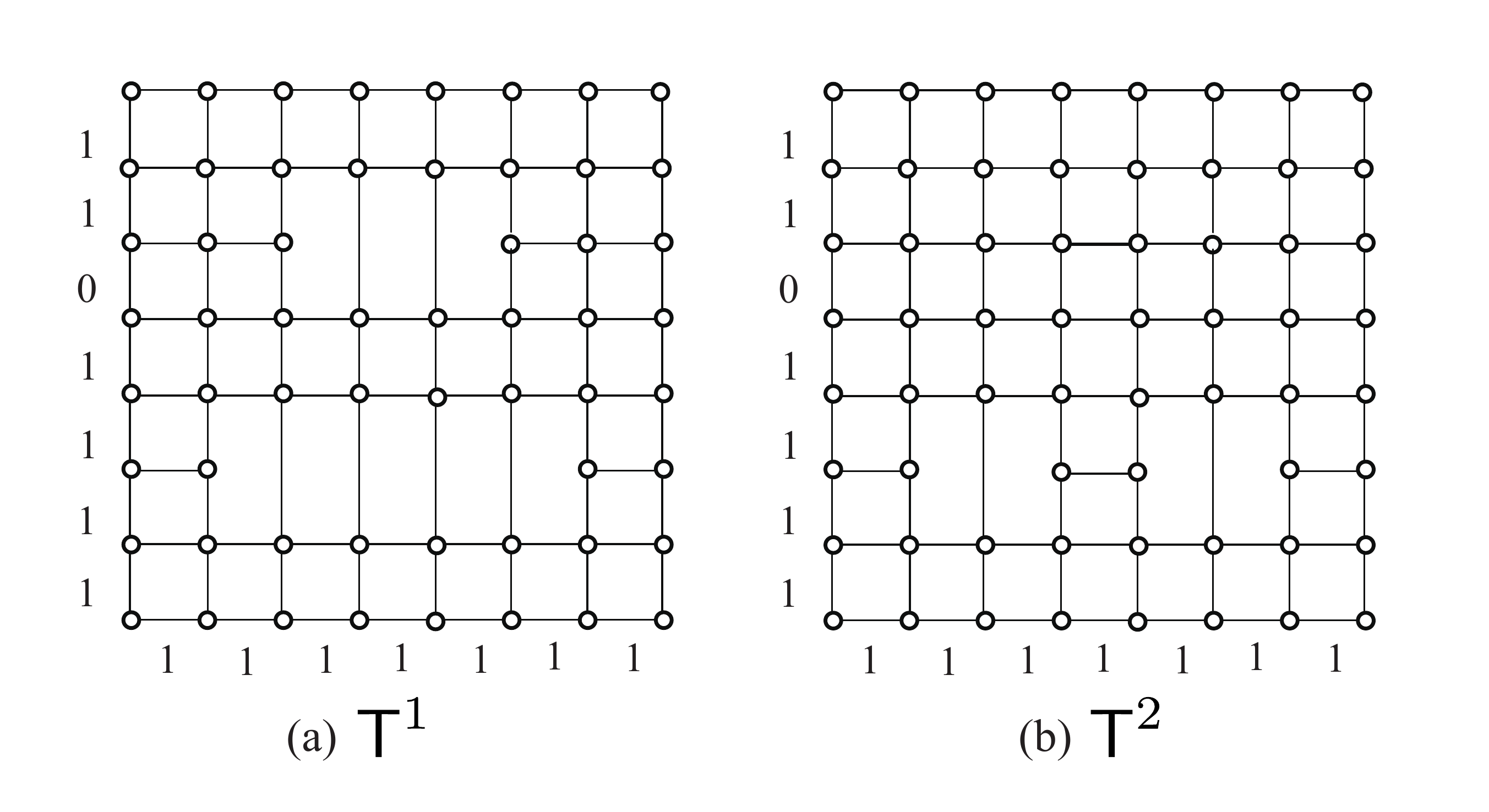}}\\
\subfigure {\includegraphics [width=4in]{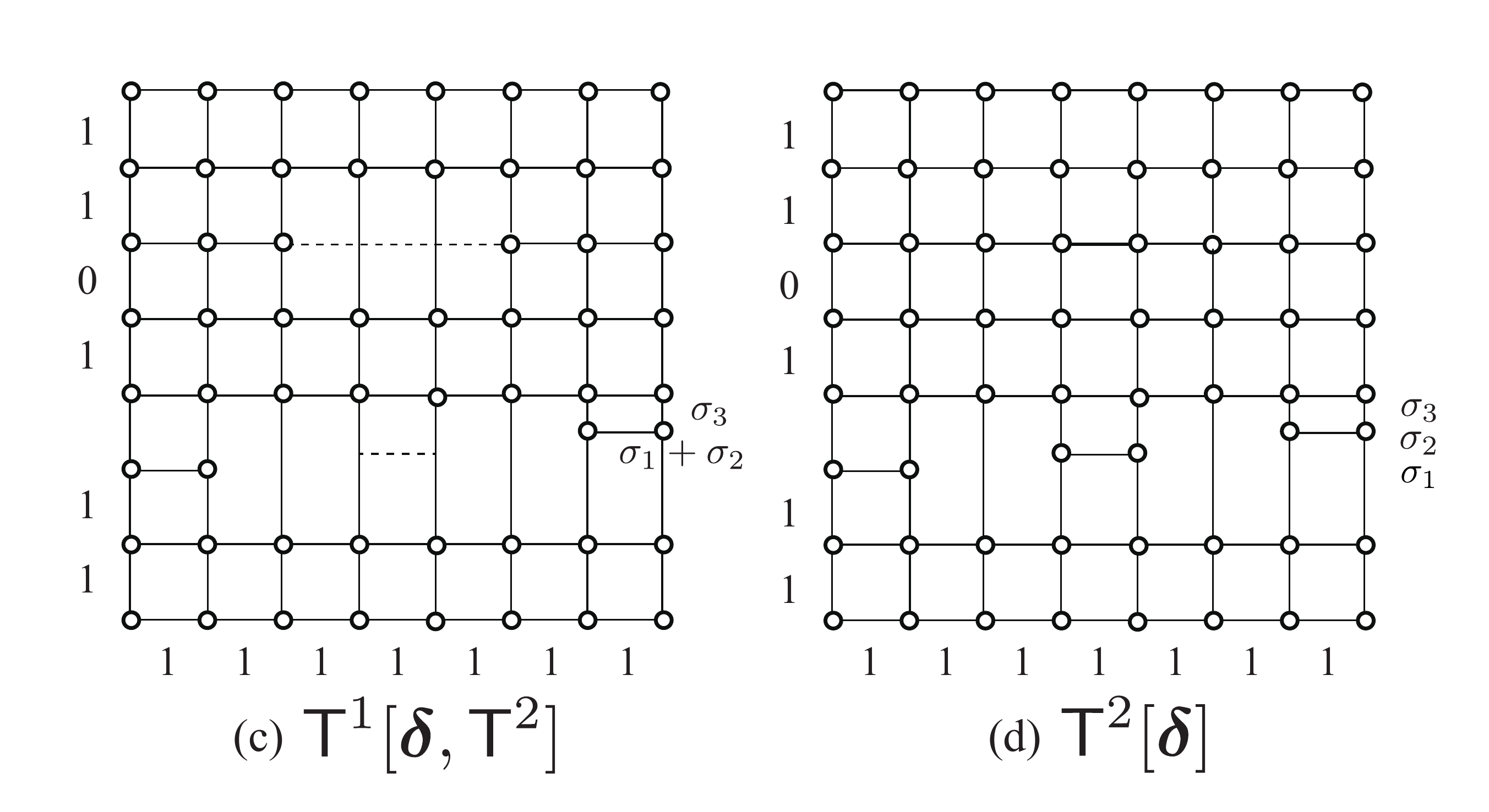}}
\caption{The construction of $\mathsf{T}^1[\boldsymbol{\delta},
  \mathsf{T}^2]$. Two analysis-suitable T-meshes are shown in (a) and (b) such
  that $\mathsf{T}^1 \subseteq \mathsf{T}^2$. The perturbed T-mesh
  $\mathsf{T}^1[\boldsymbol{\delta}, \mathsf{T}^2]$, shown in (c), is formed by
  removing the dotted lines from $\mathsf{T}^2[\boldsymbol{\delta}]$,
  shown in (d).}
\label{fig:submesh}
~\\[-1ex]
\end{figure}

\begin{definition}
Given a T-mesh, $\mathsf{T}$, with no knot multiplicities, and the
corresponding extended T-mesh, $\mathsf{T}_{ext}$, the homogeneous
extended spline space is defined as
\begin{equation}
\mathcal{S}_{ext} = \left\{ f \in C^{2,2}(\mathbb{R}^2) \, | \,  f
  |_{\tilde{Q}_{ext}} \in \mathbb{P}_{33},
 \forall \tilde{Q}_{ext} \subseteq \tilde{\Omega}, \,
 \mathrm{and} \, f |_{\mathbb{R}^2 \setminus \tilde{\Omega}} \equiv 0 \right\}
\end{equation}
where $C^{2,2}(\mathbb{R}^2)$ is the space of bivariate functions
which are $C^2$-continuous in $\xi$ and $\eta$ over all of
$\mathbb{R}^2$. $\mathbb{P}_{33}$ is the space of bicubic
polynomials. The extended spline space is defined to be
$\mathcal{T}_{ext} = \mathcal{S}_{ext} |_{\hat{\Omega}}$.
\label{def:hom_ext_space}
\end{definition}

\begin{proposition}
$\dim \mathcal{T}_{ext} = \dim \mathcal{S}_{ext}$
\label{thm:eqv_spaces}
\end{proposition}
\begin{proof}
We first prove that the dimension of $\mathcal{S}_{ext}$ is not less
than the dimension of $\mathcal{T}_{ext}$. Notice that for any
function $f \in \mathcal{S}_{ext}$, $f|_{\hat{\Omega}} \in
\mathcal{T}_{ext}$. We now show that the dimension of
$\mathcal{T}_{ext}$ is not less than the dimension of
$\mathcal{S}_{ext}$. This is equivalent to showing that there is
only one function in $\mathcal{S}_{ext}$ which is zero over
$\hat{\Omega}$. It is easy to see that the only function which is
zero over $\hat{\Omega}$ must be zero over all of $\mathbb{R}^2$ since the
minimum support of a cubic $C^2$ spline function is four intervals.
\end{proof}

\begin{lemma}
\label{lemma:charac}
If the extended T-mesh, $\mathsf{T}_{ext}$, of an
analysis-suitable T-mesh, $\mathsf{T}$, has no knot multiplicities or
overlap vertices, then $\mathcal{T}=\mathcal{T}_{ext}$.
In other words, the analysis-suitable T-spline space, $\mathcal{T}$,
and the extended spline space, $\mathcal{T}_{ext}$, are the same space.
\end{lemma}
\begin{proof}
We have that $\mathcal{T} \subseteq \mathcal{T}_{ext}$
(see~\cite{BeBuChSa12}, Lemma 4.3), so the dimension of $\mathcal{T}$ is less than
that of $\mathcal{T}_{ext}$, which, according to
Theorem~\ref{lemma:dim0}, Proposition~\ref{thm:eqv_spaces}, and
Theorem~\ref{lemma:fullrank}
is the number of active vertices. Since the blending functions for
anaysis-suitable T-splines are linearly independent the dimension of
$\mathcal{T}$ is also the number of active vertices. Thus, the
two spline spaces are identical.
\end{proof}

\begin{lemma}
\label{th:nest_result}
Given two analysis-suitable T-meshes,
$\mathsf{T}^1$ and $\mathsf{T}^2$, neither of which has knot multiplicities
or overlap vertices, if $\mathsf{T}_{ext}^1 \subseteq
\mathsf{T}_{ext}^2$, then $\mathcal{T}^1 \subseteq \mathcal{T}^2$.
\end{lemma}
\begin{proof}
Obviously, $\mathcal{T}_{ext}^1 \subseteq \mathcal{T}_{ext}^2$. Since
$\mathsf{T}^1$ and $\mathsf{T}^2$ are analysis-suitable, according to
Lemma~\ref{lemma:charac}, $\mathcal{T}^1 \subseteq \mathcal{T}^2$.
\end{proof}

\begin{lemma}
\label{lemma:in}
Given two analysis-suitable T-spline spaces, $\mathcal{T}^1$ and $\mathcal{T}^2$,
if $\mathcal{T}^1[\boldsymbol{\delta}] \subseteq
\mathcal{T}^2[\boldsymbol{\delta}]$, then
$\mathcal{T}^1 \subseteq \mathcal{T}^2$.
\end{lemma}
\begin{proof}
Suppose the perturbed T-spline space, $\mathcal{T}^{1}[\boldsymbol{\delta}]$,
is spanned by the basis functions, $N_{A}^1[\boldsymbol{\delta}]$, and
the perturbed T-spline space,
$\mathcal{T}^{2}[\boldsymbol{\delta}]$, is spanned by the
basis functions, $N_{B}^2[\boldsymbol{\delta}]$. We have that
\begin{equation*}
N_{A}^1[\boldsymbol{\delta}] = \sum
\lambda_B(N_A[\boldsymbol{\delta}]) N_{B}^2[\boldsymbol{\delta}].
\end{equation*}
where $\lambda_B$, $B \in \mathsf{A}(\mathsf{T}^2[\boldsymbol{\delta}])$ are the dual
functionals for the analysis-suitable T-spline basis $N_B[\boldsymbol{\delta}]$ as
described in~\cite{BeBuChSa12}. According to Theorem~\ref{thm:converge_basis},
\begin{equation*}
\lim_{\delta \rightarrow 0} N_{B}^2[\boldsymbol{\delta}] = N_{\pi(B)}^2, \quad
\lim_{\delta \rightarrow 0} N_{A}^1[\boldsymbol{\delta}] = N_{\pi(A)}^1.
\end{equation*}
According to Theorem 4.41 in~\cite{Schu93},
$\lambda_B(N_A[\boldsymbol{\delta}])$ is bounded, so
\begin{equation*}
\lim_{\delta \rightarrow 0}\lambda_B(N_A[\boldsymbol{\delta}]) = c_{B}^{A},
\end{equation*}
i.e.,
\begin{equation*}
N_{\pi(A)}^1 = \sum c_{B}^{A}N_{\pi(B)}^2.
\end{equation*}
\end{proof}

\begin{theorem}
\label{th:nest_result} Given two analysis-suitable T-meshes,
$\mathsf{T}^1$ and $\mathsf{T}^2$, if $\mathsf{T}^{1}_{ext}[\boldsymbol{\delta},\mathsf{T}^2]
\subseteq \mathsf{T}^{2}_{ext}[\boldsymbol{\delta}]$, then $\mathcal{T}^1 \subseteq \mathcal{T}^2$.
\end{theorem}
\begin{proof}
Obviously, $\mathcal{T}_{ext}^1[\boldsymbol{\delta},\mathsf{T}^2] \subseteq
\mathcal{T}_{ext}^2[\boldsymbol{\delta}]$.
Since $\mathsf{T}_{ext}^2[\boldsymbol{\delta}]$ has no knot multiplicities or overlap
vertices we have that $\mathcal{T}^2[\boldsymbol{\delta}] =
\mathcal{T}^2_{ext}[\boldsymbol{\delta}]$ according to
Lemma~\ref{lemma:charac}. Since the extended T-mesh,
$\mathsf{T}^1_{ext}[\boldsymbol{\delta}, \mathsf{T}^2]$, may have overlap vertices
but no knot multiplicities we have that
$\mathcal{T}^1[\boldsymbol{\delta}, \mathsf{T}^2] \subseteq
\mathcal{T}^1_{ext}[\boldsymbol{\delta}, \mathsf{T}^2]$ according to
Theorem~\ref{lemma:dim0} and Proposition~\ref{thm:eqv_spaces}. This
immediately implies that $\mathcal{T}^1 \subseteq \mathcal{T}^2$.
\end{proof}

\begin{proposition}
\label{the:opti}
Every analysis-suitable T-spline space contains the space of bicubic
polynomials.
\end{proposition}
\begin{proof}
We have that any bicubic polynomial $f[\boldsymbol{\delta}] \in
\mathcal{T}[\boldsymbol{\delta}]$ according to Lemma~\ref{lemma:charac}. Using the
dual basis for $\mathcal{T}[\boldsymbol{\delta}]$ we have that
\begin{equation*}
f[\boldsymbol{\delta}] = \sum \lambda_A(f) N_A[\boldsymbol{\delta}].
\end{equation*}
Thus, according to Theorem~\ref{thm:converge_basis} and (\cite{Schu93},
Theorem 4.41), as $\delta \rightarrow 0$,
\begin{equation*}
f = \sum c_A N_A.
\end{equation*}
\end{proof}

\begin{corollary}
\label{cor:pu}
Every analysis-suitable T-spline space forms a partition of unity. In
other words, $\sum_A N_A(\xi, \eta) = 1$, $\forall (\xi,\eta) \in \hat{\Omega}$.
\end{corollary}
\begin{proof}
This immediately follows from Proposition~\ref{the:opti} and the fact
that $\lambda_A(1) = 1$.
\end{proof}

\section{Approximation}
\label{sec:approx}
As described in~\cite{BeBuChSa12,BeBuSaVa12} approximation
properties of analysis-suitable T-splines are directly linked to
Proposition~\ref{the:opti}. In other words, having the bicubic
polynomials in the T-spline space is the minimal requirement to obtain an~$O(h^4)$
convergence rate in the mesh size.

Following the approach in~\cite{BaBeCoHuSa06,BeBuChSa12,BeBuSaVa12}, the dual
basis for an analysis-suitable T-spline space, $\mathcal{T}$,
can be used to construct a projection operator, $\mathbb{P}: L^2(\hat{\Omega})
\rightarrow \mathcal{T}$, where
\begin{equation*}
\mathbb{P}(f)(\xi, \eta) = \sum_{A \in \mathsf{A}(\mathsf{T})}
\lambda_A(f) N_A(\xi,\eta) \quad \forall f \in L^2(\hat{\Omega}),
\forall (\xi, \eta) \in \hat{\Omega}.
\end{equation*}
We denote the open support of a T-spline basis function by $Q_A
\subset \tilde{\Omega}$, and
the extended support of an element $\hat{Q}$ by $\Omega_{\hat{Q}}
\subset \tilde{\Omega}$,
where
\begin{equation*}
\Omega_{\hat{Q}} = \bigcup_{A \in \mathsf{A}(\hat{Q})} Q_A, \quad
\mathsf{A}(\hat{Q}) = \{A \in \mathsf{A}(\mathsf{T}) : Q_A \cap
\hat{Q} \neq \emptyset \}.
\end{equation*}
We will denote by $R(\Omega_{\hat{Q}})$ the smallest rectangle in
$\tilde{\Omega}$ containing $\Omega_{\hat{Q}}$ and .

\begin{proposition}
Given an analysis-suitable T-spline space, $\mathcal{T}$, the
projection operator $\mathbb{P}$ is (locally) h--uniformly continuous
in the $L^2$ norm. In other words, there exists a constant $C$
independent of $\mathsf{T}, \Xi, \Pi$ such that
\begin{equation*}
||\mathbb{P}(f)||_{L^2(\hat{Q})} \leq C ||f||_{L^2(\Omega_{\hat{Q}})}
\quad \forall \hat{Q} \in \mathsf{T}_{ext}, \forall f \in L^2(\hat{\Omega})
\end{equation*}
Note that the constant $C$ may depend on the polynomial degree.
\end{proposition}
\begin{proof}
The result follows immediately from (\cite{BeBuChSa12}, Proposition
5.4) and Proposition~\ref{the:opti}.
\end{proof}

\begin{proposition}
Given an analysis-suitable T-spline space, $\mathcal{T}$, there exists a constant $C'$
independent of $\mathsf{T}, \Xi, \Pi$ such that for $r \in [0,4]$
\begin{equation*}
||f - \mathbb{P}(f)||_{L^2(\hat{Q})} \leq C'(h_{R(\Omega_{\hat{Q}})})^r |f|_{H^r(R(\Omega_{\hat{Q}}))}
\quad \forall \hat{Q} \in \mathsf{T}_{ext}, \forall f \in L^2(\hat{\Omega})
\end{equation*}
where $h_{R(\Omega_{\hat{Q}})}$ denotes the diameter of
$R(\Omega_{\hat{Q}})$. Note that the constant $C'$ may depend on the
polynomial degree.
\end{proposition}
\begin{proof}
The result follows immediately from (\cite{BeBuChSa12}, Proposition
5.4) and Proposition~\ref{the:opti}.
\end{proof}

\section{Dimension}
\label{sec:lemma}
In this section, we develop a dimension formula for polynomial spline
spaces defined over the extended T-mesh in the parametric domain of a
T-spline and establish the
connection between this dimension formula and analysis-suitable
T-spline spaces. The dimension formula, written only in terms of
topological quantities of the original T-mesh,
is an essential ingredient in establishing the refineability
properties in Section~\ref{sec:refine} and the approximation results
in Section~\ref{sec:approx} for analysis-suitable T-splines. The
essential results are proven in
Theorems~\ref{lemma:dim0} and~\ref{lemma:fullrank}.

Unlike existing approaches, our dimension formula does not require
that the T-mesh have any nesting structure. Of critical importance is
how this dimension formula can be directly related to
analysis-suitable T-spline spaces which can then be used to construct a
simple set of basis functions for the spline space which are compatible
with commercial CAD and analysis frameworks.

\subsection{Smoothing cofactor-conformality method}
\label{sec:cofac} We use the smoothing cofactor-conformality
method~\cite{Wa01,Schu93} to transform the
smoothness properties of $\mathcal{S}_{ext}$ into a
linear constraint matrix, $\mathbf{M}$. This constraint matrix is
then analyzed to determine the dimension of $\mathcal{S}_{ext}$. We
recall that the spline space, $\mathcal{S}_{ext}$, is defined using the
extended T-mesh, $\mathsf{T}_{ext}$, corresponding to a T-mesh,
$\mathsf{T}$, which \textit{does not} have any knot multiplicities.

\subsubsection{Vertex and edge cofactors} As shown in Figure
\ref{fig:smoothv}, for any vertex, $V_{i, j} = (\xi_i, \eta_j)
\in \mathsf{T}_{ext}$, the surrounding bicubic polynomial patches are
labeled, $p_{i, j}^{k}(\xi, \eta)$, $k = 0, 1, 2, 3$. If the vertex,
$V_{i, j}$, is a T-junction, then $p_{i, j}^{k}(\xi, \eta) = p_{i,
j}^{k+1}(\xi, \eta)$ for some $k$. Since $p_{i, j}^{0}(\xi, \eta)$
and $p_{i, j}^{1}(\xi, \eta)$ are $C^{2}$-continuous there exists a
cubic polynomial $\gamma_{i,j}^{2}(\eta)$, called the \textit{edge
cofactor}, such that
\begin{equation}
p_{i, j}^{1}(\xi, \eta) - p_{i, j}^{0}(\xi, \eta) = \gamma_{i, j}^{2}(\eta)(\xi -
\xi_{i})^{3}.
\label{eq:ecofac1}
\end{equation}
\begin{figure}[htbp]
\begin{center}
\includegraphics[width=0.9\textwidth]{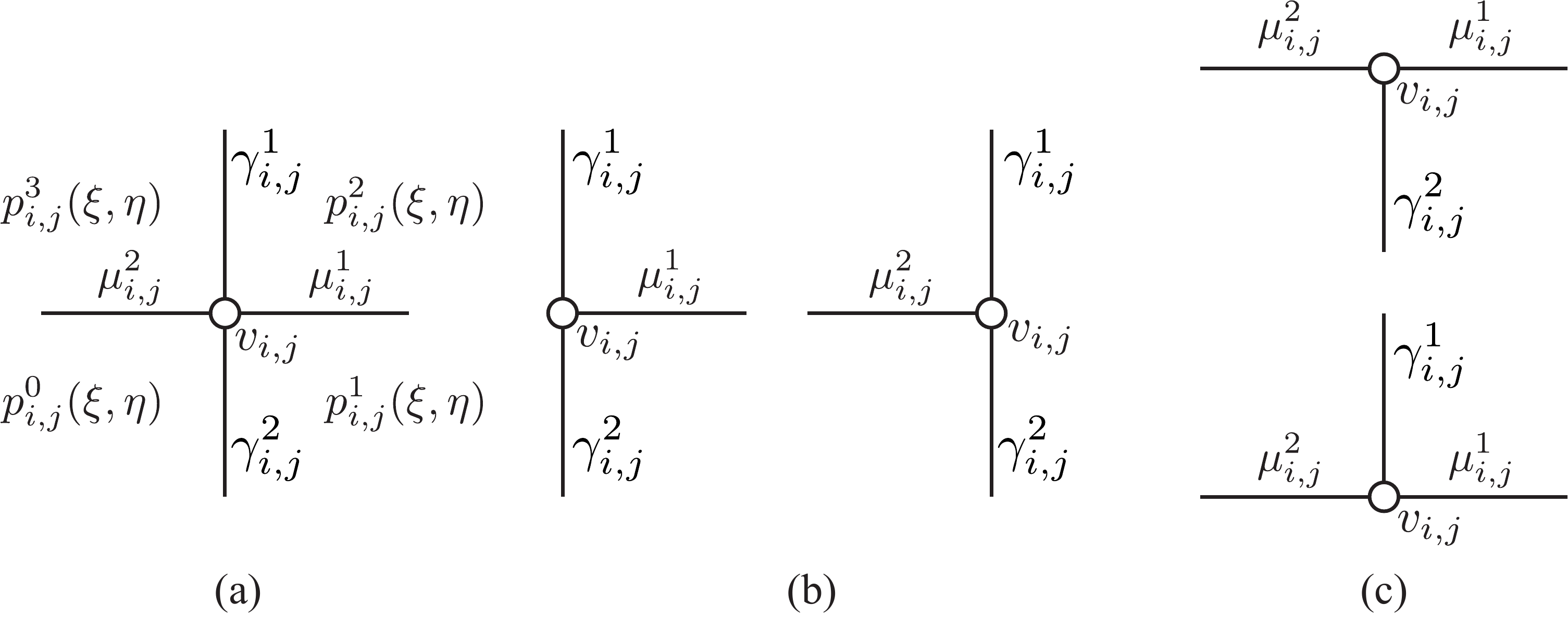}
\end{center}
\caption{The smoothing cofactors around a vertex.\label{fig:smoothv}}
\end{figure}

Similarly, there exists cubic polynomials, $\gamma_{i,
j}^{1}(\eta)$, $\mu_{i, j}^{1}(\xi)$, and $\mu_{i, j}^{2}(\xi)$,
such that
\begin{align}
p_{i, j}^{2}(\xi, \eta) - p_{i, j}^{1}(\xi, \eta) &= \mu_{i, j}^{1}(\xi)(\eta -
\eta_{j})^{3},  \label{eq:ecofac2} \\
p_{i, j}^{3}(\xi, \eta) - p_{i, j}^{2}(\xi, \eta) &= -\gamma_{i, j}^{1}(\eta)(\xi -
\xi_{i})^{3}, \label{eq:ecofac3} \\
p_{i, j}^{0}(\xi, \eta) - p_{i, j}^{3}(\xi, \eta) &= -\mu_{i, j}^{2}(\xi)(\eta -
\eta_{j})^{3}.
\label{eq:ecofac4}
\end{align}
We note that if two patches are identical the edge cofactor is zero.
Combining~\eqref{eq:ecofac1} - \eqref{eq:ecofac4} gives
\begin{equation}
(\gamma_{i, j}^{1}(\eta) - \gamma_{i, j}^{2}(\eta))(\xi - \xi_{i})^{3} =
(\mu_{i, j}^{1}(\xi) - \mu_{i, j}^{2}(\xi))(\eta - \eta_{j})^{3}.
\end{equation}
Since $(\xi - \xi_{i})^{3}$ and $(\eta - \eta_{j})^{3}$ are prime to each
other there exists a constant, $d_{i, j}$, called the \textit{vertex
  cofactor}, such that
\begin{equation}
\label{equ:edge} \gamma_{i, j}^{1}(\eta) - \gamma_{i, j}^{2}(\eta) =
d_{i, j}(\eta - \eta_{j})^{3}, \quad \mu_{i, j}^{1}(\xi) - \mu_{i, j}^{2}(\xi)
= d_{i, j}(\xi - \xi_{i})^{3}.
\end{equation}

\subsubsection{Assembling the constraint matrix, $\mathbf{M}$}

\begin{figure}[htbp]
\begin{center}
\includegraphics[width=0.9\textwidth]{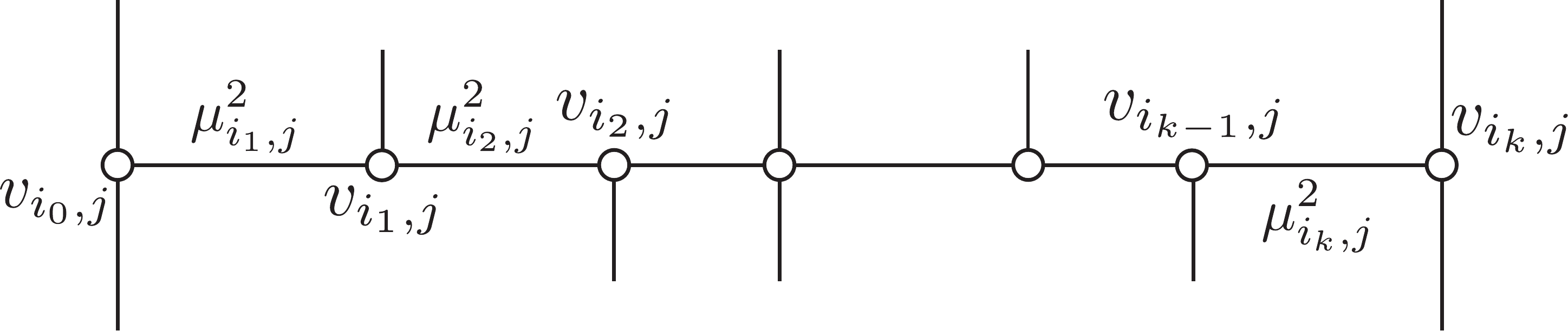}
\end{center}
\caption{The smoothing cofactors along a horizonal edge
segment.\label{fig:smoothe}}
\end{figure}
Referring to Figure \ref{fig:smoothe}, consider a horizontal segment
$G_{j}^{h}$ with $k + 1$ vertices and $k$ edge cofactors.
Using~\eqref{equ:edge} we have that
\begin{align}
\mu_{i_{0}, j}^{1} - 0 &= d_{i_{0}, j}( \xi - \xi_{i_{0}})^{3}, \label{eq:vcofac_const1}\\
\mu_{i_{1}, j}^{2} - \mu_{i_{1}, j}^{1} &= d_{i_{1}, j}( \xi -
\xi_{i_{1}})^{3},  \label{eq:vcofac_const2}\\
&\,\,\,\vdots \nonumber\\
0 - \mu_{i_{k}, j}^{2} &= d_{i_{k}, j}( \xi - \xi_{i_{k}})^{3},
\label{eq:vcofac_const3}
\end{align}
and
\begin{equation}
\mu_{i_{\ell+1}, j}^{2} = \mu_{i_{\ell}, j}^{1}, \, \, \ell = 1,
\ldots, k.
\label{eq:vcofac_const4}
\end{equation}
Summing \eqref{eq:vcofac_const1} - \eqref{eq:vcofac_const3} and using
\eqref{eq:vcofac_const4} results in the linear system
\begin{equation}
\label{equ:con0}
L_j^h := \sum_{\ell=0}^{k} d_{i_{\ell}, j}( \xi -
\xi_{i_{\ell}})^{3} = 0.
\end{equation}
We call the solution space, denoted by $W[G_{j}^{h}]$, for this linear system the \emph{edge
conformality space}. Similarly, for a vertical segment $G_i^v$ we have
that
\begin{equation}
\label{equ:con1}
L_i^v := \sum_{\ell=0}^{l} d_{i, j_{\ell}}( \eta -
\eta_{j_{\ell}})^{3} = 0
\end{equation}
where the solution space is denoted by $W[G_{i}^{v}]$.
By~\eqref{equ:con0} and~\eqref{equ:con1}, one immediately has that

\begin{lemma}
\label{lemma:dim}
If each $\xi_{i_{\ell}}$ and $\eta_{j_{\ell}}$
are different, then the dimension of $W[G_{j}^{h}]$ and
$W[G_{i}^{v}]$ are $k - 3$ and $l - 3$ respectively.
\end{lemma}

The linear systems (\ref{equ:con0}) and (\ref{equ:con1}), associated
with the horizontal and vertical segments in
$\mathsf{T}_{ext}$, can be assembled into the global system
\begin{equation}
\mathbf{M} \mathbf{D} = \mathbf{0}
\label{eq:glob_conf}
\end{equation}
where $\mathbf{D} = [d_1, d_2, \ldots, d_{n^{ext}}]^T$ is a column
vector of all vertex cofactors in $\mathsf{T}_{ext}$ and
$\mathbf{M}$ is a $4n_{seg} \times n_{ext}$ real matrix. Each edge
conformality condition corresponds to a submatrix consisting of $4$
rows of $\mathbf{M}$ and each vertex cofactor corresponds to a
column of $\mathbf{M}$.
\begin{lemma}The dimension of $\mathcal{S}_{ext}$ is the
nullity of $\mathbf{M}$, i.e., the dimension is $n_{ext}$ minus the
rank of $\mathbf{M}$.
\label{lemma:null_m}
\end{lemma}
\begin{proof}
Since the continuity constraints in $\mathcal{S}_{ext}$
  have been converted into the linear system in~\eqref{eq:glob_conf}
  using the smoothing cofactor-conformality method (see~\cite{Wa01}),
  the dimension of $\mathcal{S}_{ext}$ is the dimension of the null
  space of $\mathbf{M}$, i.e., the dimension is $n_{ext}$ minus the
  rank of $\mathbf{M}$.
\end{proof}

\subsection{Simplifying the constraint matrix, $\mathbf{M}$, and
  $\mathsf{T}_{ext}$}
It is possible to simplify the constraint matrix, $\mathbf{M}$, and
the topology of the extended T-mesh in the parametric domain,
$\mathsf{T}_{ext}$, such that
the null space of $\mathbf{M}$ is undisturbed. To \textit{remove} a
vertex from $\mathsf{T}_{ext}$ means we delete the corresponding
column from $\mathbf{M}$ and to \textit{remove} a segment from
$\mathsf{T}_{ext}$ means we delete the appropriate submatrix from
$\mathbf{M}$. We form the reduced constraint matrix
$\overline{\mathbf{M}}$ by removing the eight segments and contained
vertices $[\xi_{\underline{m}}, \xi_{\overline{m}}] \times
\{\eta_{\underline{n}}\}$, $[\xi_{\underline{m}}, \xi_{\overline{m}}] \times
\{\eta_{\underline{n} + 1}\}$, $[\xi_{\underline{m}}, \xi_{\overline{m}}] \times
\{\eta_{\overline{n} - 1}\}$, $[\xi_{\underline{m}}, \xi_{\overline{m}}] \times
\{\eta_{\overline{n}}\}$, $\{\xi_{\underline{m}}\} \times
[\eta_{\underline{n}}, \eta_{\overline{n}}]$, $\{\xi_{\underline{m} + 1}\} \times
[\eta_{\underline{n}}, \eta_{\overline{n}}]$, $\{\xi_{\overline{m} - 1}\} \times
[\eta_{\underline{n}}, \eta_{\overline{n}}]$, and $\{\xi_{\overline{m}}\} \times
[\eta_{\underline{n}}, \eta_{\overline{n}}]$.
We denote the T-mesh after the removals by
$\overline{\mathsf{T}}_{ext}$ and the number of vertices and segments
in $\overline{\mathsf{T}}_{ext}$ by $\overline{n}^{ext}$ and
$\overline{n}^{G}$, respectively. Figure~\ref{fig:simple} shows
the simplified extended T-mesh $\overline{\mathsf{T}}_{ext}$ for the
extended T-mesh in Figure~\ref{fig:tmesh_extended}.
The vertices and segments which remain after the removal process have
corresponding entries in $\overline{\mathbf{M}}$.

\begin{figure}[htbp]
\begin{center}
\includegraphics[width=0.6\textwidth]{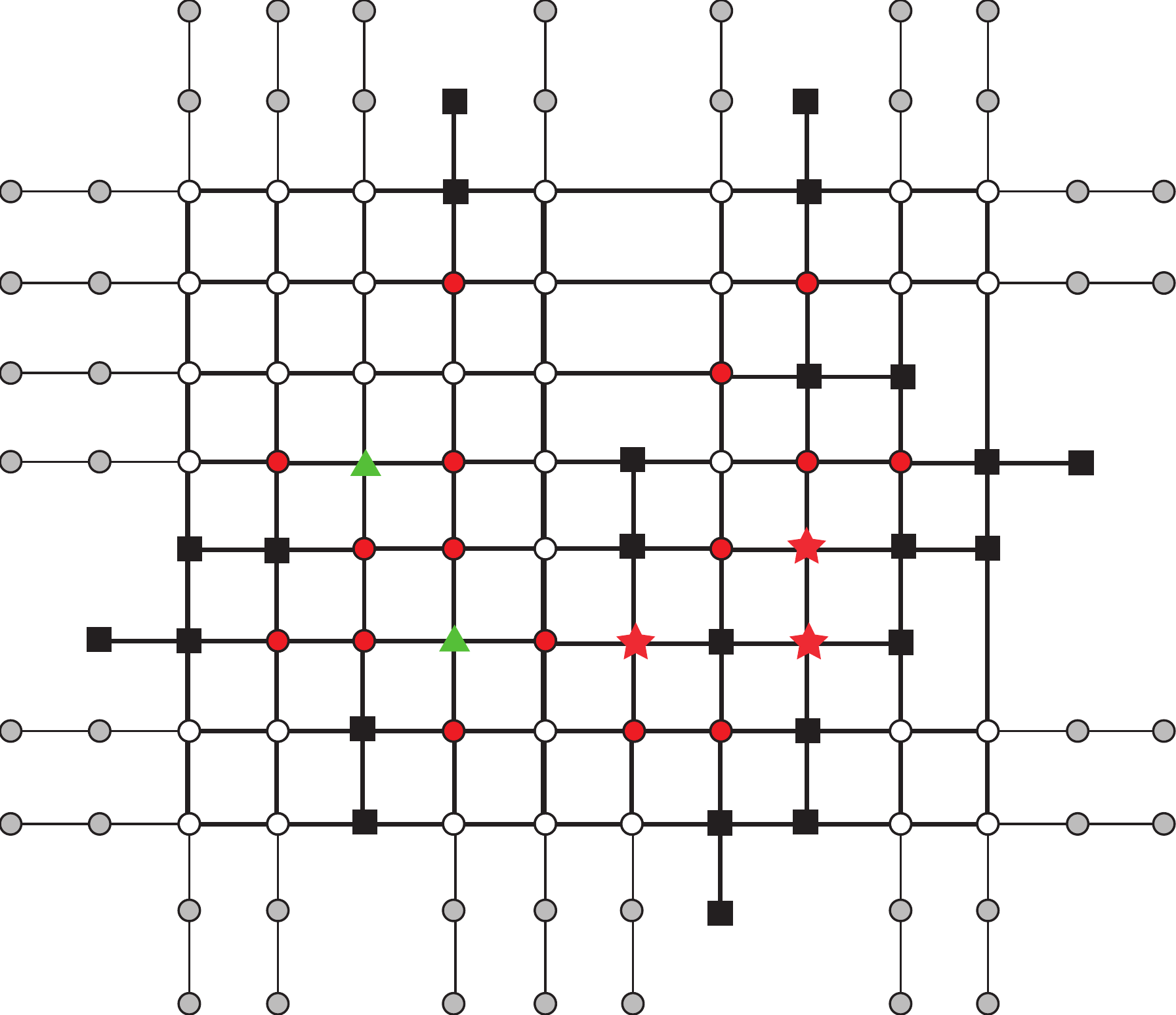}
\end{center}
\caption{The simplified extended T-mesh
  $\overline{\mathsf{T}}_{ext}$ for the extended T-mesh in
  Figure~\ref{fig:tmesh_extended}.}
\label{fig:simple}
\end{figure}

\begin{lemma}
\label{lemma:equ} The dimension of the null space of $\mathbf{M}$ is
the same as that for $\overline{\mathbf{M}}$.
\end{lemma}
\begin{proof}
The vertex cofactors which correspond to the removed corner vertices
can be uniquely determined by applying~\eqref{equ:con0} to the four
horizontal removed segments \textit{or} by applying~\eqref{equ:con1}
to the four vertical removed segments. To establish the result we
need to show that the constraints corresponding to the four vertical
removed segments can be derived from the constraints corresponding to
the four horizontal removed segments. We have that
\begin{align}
0 &= \sum_{m = 0}^{k}\left[\sum_{\ell=0}^{l} d_{i_{m}, j_{\ell}}(
\eta
- \eta_{j_{\ell}})^{3}\right](\xi - \xi_{i_{m}})^{3} \label{eq:n_reduce1}\\
 & = \sum_{\ell = 0}^{l}\left[\sum_{m=0}^{k} d_{i_{m}, j_{\ell}}( \xi -
 \xi_{i_{m}})^{3}\right](\eta - \eta_{j_{\ell}})^{3} \label{eq:n_reduce2}\\
 & = \sum_{\ell = 0, 1, l-1, l}\left[\sum_{m=0}^{k} d_{i_{m}, j_{\ell}}(
 \xi - \xi_{i_{m}})^{3}\right](\eta - \eta_{j_{\ell}})^{3}.
\label{eq:n_reduce3}
\end{align}
Equation~\eqref{eq:n_reduce1} involves the sum of all edge
conformality conditions for the horizonal edge segments.
Equation~\eqref{eq:n_reduce3} holds because the linear systems for
the other vertical segments are satisfied. Since $(\eta -
\eta_{j_{\ell}})^3$, $\ell=0,1,l-1,l$, form a basis for a linear space
of polynomials with degree less than four, $\sum_{m=0}^{k} d_{i_{m},
j_{\ell}}(
 \xi - \xi_{i_{m}})^{3} = 0$, for $m=0,\ldots,k$. In other words, the
 constraints for the four vertical removed segments can be
 derived from the other constraints.
\end{proof}

\begin{remark}
There are many possible simplification techniques which could have
been used. This simplication technique was chosen because it leaves
the null space of $\mathbf{M}$ undisturbed and each resulting segment
in $\overline{\mathbf{M}}$ contains exaclty four extended vertices.
\end{remark}

\begin{theorem}
\label{lemma:dim0}
If $\overline{\mathbf{M}}$ has full column rank, then the
dimension of $\mathcal{S}_{ext}$ is
\begin{equation} \label{dimension} \begin{array}{r@{\;}l}
\dim\mathcal{S}_{ext} = n^a + n^{+} + n^{-}
\end{array}
\end{equation}
where $n^a$ is the number of active vertices in $\mathsf{T}$ and
$n^+$ and $n^-$ are the number of crossing and overlap vertices,
respectively, in $\mathsf{T}_{ext}$.
\end{theorem}
\begin{proof}
Since there are $\overline{n}^{ext}$ and $\overline{n}^{G}$
vertices and segments, respectively, in
$\overline{\mathsf{T}}_{ext}$, $\overline{\mathbf{M}}$ is a
$4\overline{n}^{G} \times \overline{n}^{ext}$ matrix. Since
$\overline{\mathbf{M}}$ has full column rank the dimension of
$\mathcal{S}_{ext}$ is $\overline{n}^{ext} - 4\overline{n}^{G}$.
As every segment in $\overline{\mathsf{T}}_{ext}$ has exactly four
extended vertices and these four
extended vertices are not extended vertices for any other
segment, the number of extended vertices in
$\overline{\mathsf{T}}_{ext}$ is $4\overline{n}^{G}$. Thus,
\begin{equation}
\begin{array}{r@{\;}l}  \dim\mathcal{S}_{ext} =
\overline{n}^{ext} - 4\overline{n}^{G} = n^a + n^{+} + n^{-}.
\end{array}
\end{equation}
\end{proof}

\subsection{Rank of the constraint matrix $\overline{\mathbf{M}}$}
\label{sec:dim}
Since every extended vertex in
$\overline{\mathsf{T}}_{ext}$ is an extended vertex in exactly one
segment, the matrix $\overline{\mathbf{M}}$ has more columns than
rows, i.e., $\overline{n}^{ext} > 4\overline{n}^{G}$. After
arranging the order of edge conformality conditions and the order of
vertex cofactors, an appropriate partition of the linear system of
constraints, $\overline{\mathbf{M}} \, \, \overline{\mathbf{D}} =
\mathbf{0}$, is
 \begin{equation}
    \left[\begin{array}{c|c}
        \overline{\mathbf{M}}_1 & \overline{\mathbf{M}}_2
      \end{array}\right]
    \left[
      \begin{array}{c}
        \overline{\mathbf{D}}_{1} \vspace{2pt} \\
        \hline \vspace{-10pt} \\
        \overline{\mathbf{D}}_2
      \end{array}\right] =
    \mathbf{0}
  \end{equation}
where $\overline{\mathbf{M}}_1$ is a $4\overline{n}^{G} \times
4\overline{n}^{G}$ matrix and $\overline{\mathbf{M}}_1$ is a
$4\overline{n}^{G} \times (\overline{n}^{ext} - 4\overline{n}^{G})$ matrix,
$\overline{\mathbf{D}}_1$ is a vector of the first
$4\overline{n}^{G}$ vertex cofactors, and $\overline{\mathbf{D}}_2$ is
a vector of the remaining vertex cofactors.

\begin{definition}
A simplified extended T-mesh is called
diagonalizable if we can arrange all the segments $G_{i}, i
= 1, \dots, \overline{n}^{G}$ such that the number of vertices
on segment $G_{j}$ but not on segment $G_{i}, i < j$ is at least 4.
\label{def:seg_order}
\end{definition}
\begin{lemma}
\label{lemma:four}
If a simplified extended T-mesh is diagonalizable, then the matrix
$\overline{\mathbf{M}}$ has full column rank.
\end{lemma}
\begin{proof}
Since the simplified extended T-mesh is diagonalizable, the segments
can be ordered as described in Definition~\ref{def:seg_order}. Given
this ordering of segments, for $i = 1, \ldots, \overline{n}^{G}$, we
place the edge conformality conditions corresponding to segment $G_i$
in rows $4(\overline{n}^{G}-i) + 1$ through $4(\overline{n}^{G}-i) +4$ of
$\overline{\mathbf{M}}$, and place the first four vertex cofactors which
appear in $G_{i}$ but not in $G_{j}, j < i$ in columns
$4(\overline{n}^{G}-i) + 1$ through $4(\overline{n}^{G}-i) +4$ of
$\overline{\mathbf{D}}_1$. Then the matrix $\overline{\mathbf{M}}_1$
is in upper block triangular form and according to
Lemma~\ref{lemma:dim} each diagonal block $4\times4$ matrix is full
rank, thus matrix $\overline{\mathbf{M}}_1$ is obviously of full
rank.
\end{proof}

\begin{lemma}\label{lemma:dia}
$\overline{\mathsf{T}}_{ext}$ is diagonalizable if and only if for any set of segments there
exists at least one segment in the set that has at least four vertices
which are not in the other segments in the set.
\end{lemma}
\begin{proof}
Assume the T-mesh is diagonalizable but there exists a set of segments
$\{G_{i_{j}}, j = 1, \dots, s\}$ such that any
segment in the set has at most three vertices which are not on the other segments
in the set. Without loss of generality, assume the diagonalizable
segment ordering for $\overline{\mathsf{T}}_{ext}$ is
$G_{i}, i = 1, \dots, \overline{n}^{G}$. Let $k$ be the maximal index
for all $i_{j}, j = 1,\dots, s$ and consider the set $\{G_{i}, i = 1,
\dots, k\}$. Since $G_{k}$ has at most three vertices which are not in
the segments $G_{i_{j}}, j = 1, \dots, s$, there exists an index $i, i
< k$ such that the number of vertices on
segment $G_k$ but not on segment $G_i$ is at most three. This violates
the assumption that $\overline{\mathsf{T}}_{ext}$ is
diagonalizable.

Suppose for any set of segments in $\overline{\mathsf{T}}_{ext}$ there
exists at least one segment which has at
least four vertices which are not on the other segments in the
set. For the set containing every segment, according to the
assumption, there exists one segment, $G_{\overline{n}^G}$, which
has at least four vertices which are not on the other segments. Now,
removing $G_{\overline{n}^G}$ from the set, we have that in the set of
remaining segments there exists one segment, $G_{\overline{n}^G - 1}$, which has at
least four vertices which are not on the other segments. Continuing this
process, we can arrange all the segments $G_{i}, i = 1, 2, \dots,
\overline{n}^{G}$ such that the number of vertices
on segment $G_{j}$ but not on segment $G_{i}, i < j$ is at least
4. Thus $\overline{\mathsf{T}}_{ext}$ is diagonalizable.
\end{proof}

\begin{theorem}
\label{lemma:fullrank}
For an analysis-suitable T-mesh, matrix $\overline{\mathbf{M}}$ has full column rank.
\end{theorem}
\begin{proof} For an analysis-suitable T-mesh we will prove that
the corresponding simplified extended T-mesh is diagonalizable.
Otherwise, according to Lemma~\ref{lemma:dia}, there exists a set of
segments such that each segment in the set has at most three vertices
which are not on the other segments in the set. It is evident that
the set must contain horizonal segments. Otherwise, any vertical
segment violates the assumption because it must have at least
four vertices which are not in the other segments in the set. Let $G_{i}$ be the
bottommost horizonal segment in the set (if there is more than one
such segment choose one of them). Since $G_{i}$ has at most three
vertices which are not on the other segments in the set, one of the four
extended vertices of $G_{i}$ must lie on a vertical segment,
$G_{j}$. Now, referring to Figure~\ref{fig:dim_proof}, since the T-mesh is
analysis-suitable, $G_j$ must have two anchor vertices whose vertical
index coordinate is less than $a$. Otherwise, the T-mesh is not
analysis-suitable due to intersecting T-juction
extensions. Additionally, $G_j$ must have two extended vertices whose
vertical index coordinate is less than $a$. Thus, there are four
vertices in $G_j$ which do not belong to any other segment in the set
which contradicts the assumption.
\end{proof}

\begin{figure}[htbp]
\begin{center}
\includegraphics[width=0.75\textwidth]{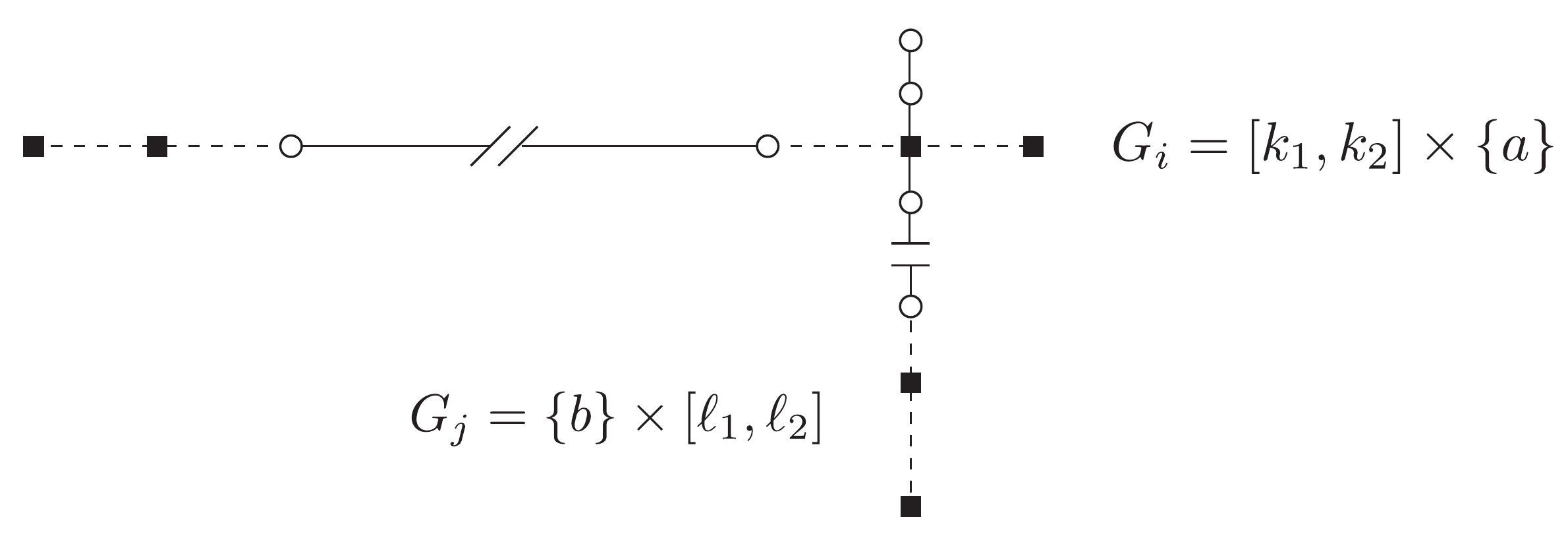}
\end{center}
\caption{A schematic for
Theorem~\ref{lemma:fullrank}.\label{fig:dim_proof}}
\end{figure}

\section{Conclusion}
\label{sec:conc}
We have established several important properties of analysis-suitable
T-splines. We developed a characterization of ASTS spaces and proved
that the space of bicubic polynomials, defined over the extended
T-mesh of an ASTS, is contained in the
corresponding ASTS space. We then proved the
conditions under which two ASTS spaces are
nested. This provides the theoretical foundation for the
analysis-suitable local refinement algorithm
in~\cite{ScLiSeHu10}. Using the characterization of ASTS we then
proved several basic approximation results. Additionally, we developed
the theory of perturbed ASTS and a simple mesh-based dimension formula
which is written in terms of the vertices in the extended T-mesh of a
T-spline. Both of these developments were critical for the proofs in this
paper and may have important applications in other contexts. While the
developments in this paper are restricted to bicubic surfaces the
extension to arbitrary degree should be straightforward.

\section*{Acknowledgements}
The authors thank Prof. A. Buffa, G. Sangalli, R. Vazquez,
University of Pavia, and T. W. Sederberg, Brigham Young University,
for several fruitful and insightful discussions on the technical
aspects of this work. This work was supported by grants from the NSF of China
(Nos.11031007, 60903148), the Chinese Universities Scientific Fund,
SRF for ROCS SE, and the Chinese Academy of Science (Startup
Scientific Research Foundation).  M.A. Scott was partially supported
by an ICES CAM Graduate Fellowship. This support is gratefully
acknowledged.

\appendix

\section{Notational conventions}
\label{sec:notation}
In Tables~\ref{tab:2} and~\ref{tab:3} the notational conventions used
throughout the text are listed as well as the section where they are
defined.
\begin{table}[htbp]
  \centering
\caption{Notational conventions used throughout the text and where
  they are defined.\label{tab:2}}
  \begin{tabular}{ l l l}
  \hline
  Symbols & Description & Section \\
  \hline
  $\mathsf{T}$ & A T-mesh & \ref{sec:tmesh_def} \\
  $\mathsf{V}$ & All vertices in a T-mesh & \ref{sec:tmesh_def} \\
  $\mathsf{E}$ & All edges in a T-mesh & \ref{sec:tmesh_def} \\
  $\mathsf{S}$ & The skeleton of a T-mesh &  \ref{sec:tmesh_def} \\
  $h\mathsf{S}$ $(v\mathsf{S})$ & The horizontal (vertical) skeleton of a T-mesh & \ref{sec:tmesh_def} \\
  $Q$ & An element of a T-mesh & \ref{sec:tmesh_def} \\
  $\mathsf{G}$ & All segments in a T-mesh & \ref{sec:segments} \\
  $n^G$ & The number of segments in a T-mesh & \ref{sec:segments} \\
  $h\mathsf{G}$ ($v\mathsf{G}$) & All horizontal (vertical) segments in a T-mesh & \ref{sec:segments} \\
  $h\mathsf{G}(a)$ $(v\mathsf{G}(a))$ & The horizontal (vertical)
  segments at vertical (horizontal) index $a$ & \ref{sec:segments} \\
  $\mathsf{R}$ & The index domain of a T-mesh & \ref{sec:tmesh_def} \\
  $\mathsf{AR}$ & The active region of an index domain & \ref{sec:tmesh_def} \\
  $\mathsf{FR}$ & The frame region of an index domain & \ref{sec:tmesh_def} \\
  \hline
  $\mathsf{T}_{ext}$ & An extended T-mesh & \ref{sec:ext_tmesh} \\
  $ext^{f}(\mathsf{T})$ & All face extensions in $\mathsf{T}$ & \ref{sec:ext_tmesh} \\
  $ext^f(T)$ & The face extension of T-junction $T$ &
  \ref{sec:ext_tmesh} \\
  $ext^{e}(\mathsf{T})$ & All edge extensions in $\mathsf{T}$ &  \ref{sec:ext_tmesh} \\
  $ext^e(T)$ & The edge extension of T-junction $T$ &\ref{sec:ext_tmesh} \\
  $ext(\mathsf{T})$ & All extensions in $\mathsf{T}$ & \ref{sec:ext_tmesh} \\
  $ext(T)$ & The extension of T-junction $T$ & \ref{sec:ext_tmesh} \\
  $hext^f(\mathsf{T})$ $vext^f{\mathsf{T}}$ & All horizontal (vertical) faces extensions in a T-mesh & \ref{sec:ext_tmesh} \\
  $\mathsf{CV}$ & All crossing vertices in $\mathsf{T}_{ext}$ & \ref{sec:ext_tmesh} \\
  $n^+$ & The number of crossing vertices in $\mathsf{T}_{ext}$ & \ref{sec:ext_tmesh} \\
  $\mathsf{OV}$ & All overlap vertices in $\mathsf{T}_{ext}$ & \ref{sec:ext_tmesh} \\
  $n^-$ & The number of overlap vertices in $\mathsf{T}_{ext}$ & \ref{sec:ext_tmesh} \\
  $\mathsf{EV}$ & All extended vertices in $\mathsf{T}_{ext}$ & \ref{sec:ext_tmesh} \\
  $n^*$ & The number of overlap vertices in $\mathsf{T}_{ext}$ & \ref{sec:ext_tmesh} \\
  $n^{ext}$ & The number of vertices in $\mathsf{T}_{ext}$ &
  \ref{sec:ext_tmesh} \\
  \hline
\end{tabular}
\end{table}

\begin{table}[htbp]
  \centering
  \caption{Notational conventions used throughout the text and where
    they are defined.\label{tab:3}}
  \begin{tabular}{ l l l}
    \hline
    Symbols & Description & Section \\
    \hline
    $\mathsf{A}(\mathsf{T})$ & All anchors in a T-mesh & \ref{sec:anchors} \\
    $n^A$ & The number of anchors in a T-mesh & \ref{sec:anchors} \\
    $\mathsf{J}$ & The set of T-junctions in a T-mesh & \ref{sec:anchors} \\
    $hv(A)$ ($vv(A)$) &  The horizontal (vertical) local index vector for $N_A(\xi, \eta)$ & \ref{sec:ts_spaces} \\
    $\Xi_A$ ($\Pi_A$) & The horizontal (vertical) local knot vector for $N_A(\xi, \eta)$ & \ref{sec:ts_spaces} \\
    $\tilde{\Omega}$ & The full parametric domain of $\mathsf{T}$ & \ref{sec:ts_spaces} \\
    $\tilde{Q}$ & A T-spline element in the full parametric domain & \ref{sec:ts_spaces} \\
    $\hat{\Omega}$ & The reduced parametric domain of $\mathsf{T}$ & \ref{sec:ts_spaces} \\
    $\hat{Q}$ & A T-spline element in the reduced parametric domain & \ref{sec:ts_spaces} \\
    $\Omega_{\hat{Q}}$ & The extended support of $\hat{Q}$ & \ref{sec:approx} \\
    $R(\Omega_{\hat{Q}})$ & The smallest rectangle containg $\Omega_{hat{Q}}$ & \ref{sec:approx} \\
    $N_A(\xi, \eta)$ & The T-spline blending function for anchor $A$ & \ref{sec:ts_spaces} \\
    $Q_A$ & The support of $N_A$ & \ref{sec:approx} \\
    $\lambda_A$ & The dual basis function associated with $N_A$ & \ref{sec:ts_spaces} \\
    $\mathcal{T}$ & A T-spline space & \ref{sec:ts_spaces} \\
    $\mathcal{S}_{ext}$ & A homogeneous extended spline space & \ref{sec:refine} \\
    $\mathcal{T}_{ext}$ & An extended spline space & \ref{sec:refine} \\
    \hline
    $\Xi[\boldsymbol{\delta}] $($\Pi[\boldsymbol{\delta}]$) & A
    horizontal (vertical) perturbed global knot vector &
    \ref{sec:perturb} \\
    $\imath = \imath(i, g)$ ($\jmath = \imath(j, g)$) & the index of the $g^{th}$ segment in
    $h\mathsf{G}(\{i\})$ ($v\mathsf{G}(\{j\})$) & \ref{sec:perturb} \\
    $\pi$ & the map from perturbed index $\imath(i,g)$ to index $i$ & \ref{sec:perturb} \\
    $\mathsf{T}(\boldsymbol{\delta})$ & A perturbed T-mesh &
    \ref{sec:perturb} \\
    $\mathcal{T}(\boldsymbol{\delta})$ & A perturbed T-spline space &
    \ref{sec:perturb} \\
    $\mathsf{T}(\boldsymbol{\delta}, \mathsf{T}^i)$ & A perturbed T-mesh
    generated from $\mathsf{T}^i(\boldsymbol{\delta})$ & \ref{sec:perturb} \\
\hline
  \end{tabular}
\end{table}

\bibliographystyle{elsarticle-num}
\bibliography{bibliography}
\end{document}